\documentclass[12pt]{article}

\usepackage{graphicx}
\usepackage{amsmath, amssymb, amsthm, mathrsfs, mathtools, amscd}
\usepackage[latin1]{inputenc}
\usepackage{braket}
\usepackage{graphicx}
\usepackage[all]{xy}
\usepackage{color}
\usepackage[PostScript=dvips]{diagrams}
\usepackage{enumerate}
\usepackage{varwidth}
\usepackage{times}
\usepackage{cite}
\newtheorem{theorem}{Theorem}[section]

\newtheorem{definition}{Definition}[section]

\newtheorem{proposition}[theorem]{Proposition}

\newtheorem{lemma}[theorem]{Lemma}

\newtheorem{corollary}[theorem]{Corollary}

\newcommand{\bDefinition}{\begin{definition}}
\newcommand{\eDefinition}{\end{definition}}
\newcommand{\bTheorem}{\begin{theorem}}
\newcommand{\eTheorem}{\end{theorem}}
\newcommand{\bProposition}{\begin{proposition}}
\newcommand{\eProposition}{\end{proposition}}
\newcommand{\bCorollary}{\begin{corollary}}
\newcommand{\eCorollary}{\end{corollary}}
\newcommand{\bProof}{\begin{proof}}
\newcommand{\eProof}{\end{proof}}	
\usepackage[margin=1in]{geometry}
\allowdisplaybreaks

\newcommand{\benum}{\par\begin{center}\begin{varwidth}{\textwidth}\begin{enumerate}[{\rm (i)}]}
\newcommand{\benumx}[1]{\par\begin{center}\begin{varwidth}{\textwidth}\begin{enumerate}[{\rm (#1)}]}
\newcommand{\eenum}{\end{enumerate}\end{varwidth}\par\end{center}}

\newcommand{\deqs}[1]{\begin{align*}#1\end{align*}}
\newcommand{\ie}{\mbox{i.e.\ }}
\newcommand{\eg}{\mbox{e.g.\ }}
\newcommand{\Forall}{~~\mb{for all}~~}
\newcommand{\mb}{\mbox}
\newcommand{\bc}{\begin{center}}
\newcommand{\ec}{\end{center}}
\newcommand\ul[1]{\underline{#1}}
\newcommand\ld{\lambda}
\newcommand\cBc{\mathcal{B}_c}			
\newcommand\Sig{\underline{\Sigma}}	
\newcommand\SigL{\Sig^L}						
\newcommand\SigM{\Sig^M}						
\newcommand\de{\delta}							
\newcommand\De{\Delta}							
\newcommand\eps{\varepsilon}						
\newcommand\pde{\underline{\delta}}	
\newcommand\join{\vee}							
\newcommand\bjoin{\bigvee}					
\newcommand\meet{\wedge}						
\newcommand\bmeet{\bigwedge}				
\newcommand\cH{\mathcal{H}}					
\newcommand\PH{\mathcal{P}(\cH)}		
\newcommand\BcPH{\cBc(\PH)}					
\newcommand\CSet{\mathbf{Set}}			
\newcommand\BA{\mathbf{BA}}					
\newcommand\cBA{\mathbf{cBA}}				
\newcommand\Stone{\mathbf{Stone}}		
\newcommand\Stonean{\mathbf{Stonean}}	
\newcommand\OML{\mathbf{OML}}				
\newcommand\Pos{\mathbf{Pos}}				
\newcommand\Lat{\mathbf{Lat}}				
\newcommand\dom{\operatorname{dom}}	
\newcommand\func{\operatorname{func}}	
\newcommand\ran{\operatorname{ran}}	
\newcommand\op{\operatorname{op}}		
\newcommand\cl{\operatorname{cl}}		
\newcommand\Subcl{\operatorname{Sub}_{\cl}}
\newcommand\BL{\mathcal{B}(L)}			
\newcommand\BcL{\mathcal{B}_c(L)}		
\newcommand\BM{\mathcal{B}(M)}
\newcommand\tphi{\tilde\phi}
\newcommand\Presh[1]{\mathbf{Presh}(#1)}
\newcommand\pair[2]{\langle #1,#2\rangle}
\newcommand\id{\operatorname{id}}
\newcommand\BH{{\mathcal{L}(\cH)}}
\newcommand\SigPH{\Sig^{\PH}}
\newcommand\wpsi{\ps{\mathfrak{w}}^{\psi}}
\newcommand\ps[1]{\underline{#1}}		
\newcommand\cB{\mathcal{B}}
\newcommand\cC{\mathcal{C}}
\newcommand\true{\mathsf{true}}
\newcommand\One{\mathbf{1}}
\newcommand\cE{\mathcal{E}}

\newcommand{\ph}{\phi}

\newcommand{\val}[1]{\|#1\|}
\newcommand{\valj}[1]{\val{#1}_j}
\renewcommand{\And}{\wedge}
\newcommand{\Or}{\vee}
\newcommand{\Not}{\neg}
\newcommand{\Then}{\rightarrow}
\newcommand{\THEN}{\Rightarrow}
\newcommand{\THENJ}{\Rightarrow_j}
\newcommand{\Iff}{\leftrightarrow}
\newcommand{\IFF}{\Leftrightarrow}
\newcommand{\IFFJ}{\Leftrightarrow_j}

\newcommand{\Inf}{\bigwedge}
\newcommand{\Sup}{\bigvee}
\newcommand{\ck}[1]{\hat{#1}}
\newcommand{\commutes }{\,\rotatebox[origin=c]{270}{$\multimap$}\,}
\newcommand{\p}{{}^{\perp}}
\newcommand{\ep}{\varepsilon}
\newcommand{\Q}{\mathbb{Q}}
\newcommand{\R}{\mathbb{R}}
\newcommand{\rS}{{\mathrm S}}
\newcommand{\rC}{{\mathrm C}}
\newcommand{\rR}{{\mathrm R}}

\newcommand{\la}{\lambda}
\newcommand{\Sub}{\Subcl(\Sigma)}
\newcommand{\RS}{\R^{(\Sub)}}
\newcommand{\RSR}{\R^{(\Sub)}_{{\rm r}}}

\newcommand{\RPH}{\R^{(\PH)}}
\newcommand{\VS}{V^{(\Sub)}}
\newcommand{\VL}{V^{(L)}}
\newcommand{\VPH}{V^{(\PH)}}
\newcommand{\oE}{\overline{E}}
\newcommand{\tE}{\tilde{E}}
\begin{document}
\title{\textbf{A Bridge Between Q-Worlds}}
\author{Andreas D{\"o}ring, Benjamin Eva and Masanao Ozawa}
\date{}  
\maketitle

\abstract{Quantum set theory (QST) and topos quantum theory (TQT) are two long running projects in the mathematical foundations of quantum mechanics that share a great deal of conceptual and technical affinity. Most pertinently, both approaches attempt to resolve some of the conceptual difficulties surrounding quantum mechanics by reformulating parts of the theory inside of non-classical mathematical universes, albeit with very different internal logics. We call such mathematical universes, together with those mathematical and logical structures within them that are pertinent to the physical interpretation, `Q-worlds'. Here, we provide a unifying framework that allows us to (i) better understand the relationship between different Q-worlds, and (ii) define a general method for transferring concepts and results between TQT and QST, thereby significantly increasing the expressive power of both approaches. Along the way, we develop a novel connection to paraconsistent logic and introduce a new class of structures that have significant implications for recent work on paraconsistent set theory.}

\section{Introduction}

The idea that the conceptually and philosophically challenging aspects of quantum mechanics (QM) can be understood and even resolved via the adoption of some suitably non-classical logic has been an influential one. In particular, Birkhoff and von Neumann's \cite{BirvNe36} contention that the distributive law is not generally applicable to the description of quantum systems led to the emergence of quantum logic as an important research program in logic, quantum foundations and the philosophy of science.\footnote{For discussions of the philosophical aspects of quantum logic, see \eg Putnam \cite{Put75}, Dummett \cite{Dum76}, Gibbins \cite{Gib08}.}

One possibility created by the logical perspective on QM is to construct non-classical mathematical universes (\ie models of set theory, toposes, etc.) whose internal logic is non-classical and suitably `quantum', and inside of which one can reformulate parts of the theory in a novel and illuminating way. We will call such a mathematical universe, together with those internal logical and mathematical structures that are relevant for the quantum physical interpretation, a `Q-world'.

The first example of such Q-worlds arose from a result of Takeuti \cite{Tak81},  who showed that for any quantum system, there exists a set theoretic universe\footnote{Based on the quantum logic associated with that system.} whose real numbers are in bijective correspondence with the physical quantities associated with that system. These models were subsequently generalised by Ozawa and others (see \eg Ozawa \cite{Oza07,Oza16,Oza17}, Titani \cite{Tit99},
{Ying \cite {Yin05}}) to set-theoretic structures with non-distributive internal logics. The study of these structures has come to be known as `quantum set theory' (QST). The fact that the real numbers in those structures are in bijective correspondence with the set of all physical quantities associated with the given quantum system allows one to reformulate the physics of the system in the internal language of the structures. 

The second example, first studied by Isham \cite{Ish97}, is given by the topos-theoretic reformulation of quantum mechanics. The relevant foundational result here is the reformulation of the Kochen-Specker theorem as a result about the non-existence of global sections of the `spectral presheaf' of a quantum system (see \eg Isham and Butterfield \cite{IshBut98}). Building on this result, the topos-theoretic approach progresses to reformulate quantum mechanics inside of a particular type of presheaf topos. Unlike the Q-worlds studied by Takeuti, Ozawa and others, these `quantum toposes' have a distributive, intuitionistic internal logic. For some of the literature studying the technical and conceptual implications of reformulating standard Hilbert space quantum mechanics in the topos-theoretic setting, see \eg Isham and Butterfield \cite{IshBut98}, D{\"o}ring and Isham \cite{DoeIsh11}, D{\"o}ring \cite{Doe12}, and references therein. We will call this approach `topos quantum theory' (TQT for short).

Until now, the tantalising prospect of unifying these two (kinds of) Q-worlds, non-distributive set theoretic QST and distributive topos-theoretic TQT, within a single formal setting has gone almost completely unexplored (the prospect was first tentatively suggested by Eva \cite{Eva15}). In the present article we develop such a generalised framework and present a number of results that connect QST and TQT. Quite unexpectedly,\footnote{$\ldots$ at least for us, at the beginning of our explorations...} the bridge between the Q-worlds of TQT and those of QST  turns out to be via paraconsistent set theory.

The article is structured as follows. Section \ref{Sec_OrthomodQLogic} recalls some basic ideas from non-distributive quantum logic. In section \ref{Sec_QST} we provide a concise introduction to algebraic valued models of set theory in general, and to QST in particular, summarising the key results of the approach. In section \ref{Sec_TQT} we outline the key ideas of the topos-theoretic reformulation of quantum mechanics, focusing especially on the approach's distributive logical structure. Section \ref{Sec_ParaconsAndDistribInQL} contains the first original contributions of the paper. In particular, we show that embedding standard non-distributive quantum logic into the intuitionistic logic of TQT naturally results in a novel form of paraconsistent quantum logic. After proving a number of new results about the relationship between these different forms of quantum logic, we go on define a new class of paraconsistent Q-worlds in section \ref{Sec_VSubclSig}. We then prove transfer theorems that guarantee the satisfaction of large fragments of classical mathematics inside of these Q-worlds, thereby making an important contribution to the study of inconsistent mathematics and paraconsistent set theory.\footnote{As studied, for example, by Weber \cite{Web10,Web12}, McKubre-Jordens and Weber \cite{Web12}.} Section \ref{Sec_Bridge} shows how these new Q-worlds can be seen as a bridge between the Q-worlds of TQT and QST, and uses them to transfer a number of results between the two settings. Section \ref{Sec_Conclusion} concludes. 

\section{Orthomodular Quantum Logic}	\label{Sec_OrthomodQLogic}

The logical structure of classical physics can be summarised in the following way. To any classical physical system $S$ we can associate a corresponding space $\textbf{S}$ of possible physical states of $S$. A `physical proposition' pertaining to $S$ is a statement of the form ``the value of the physical quantity $A$ for $S$ lies in the Borel set $\de \subseteq \mathbb{R}$''.\footnote{For instance, the sentences ``the momentum of the system $S$ is between 0 and 1 (in suitable units)'' and ``the velocity of $S$ is less than 5'' are examples of physical propositions pertaining to $S$.} Physical propositions of this form are in bijective correspondence (modulo logical equivalence) with measurable subsets of the state space $\textbf{S}$. In particular, each physical proposition corresponds to the set of all states that make that proposition true, which is measurable (and any measurable set of states defines a corresponding physical proposition). Since the measurable subsets of $\textbf{S}$ form a complete Boolean algebra under the usual set theoretic operations (modulo sets of Lebesgue measure 0), we can conclude that the logic governing the physical propositions associated with a classical system $S$ is classical. 

The situation is very different in QM. Given a quantum system $S$, the space of states is always assumed to be a Hilbert space $\cH$, and physical quantities (\eg angular momentum, position, mass etc.) are represented not as real valued functions on $\cH$, but rather as \emph{self-adjoint operators} on $\cH$. The spectral theorem for self-adjoint operators tells us that physical propositions (statements about the values of physical quantities) are in bijective correspondence up to logical equivalence with the closed subspaces of $\cH$ or, equivalently, with projection operators onto the closed subspaces of $\cH$. Thus the logical structure of the physical propositions associated with a quantum system is given by the lattice of projection operators on the corresponding Hilbert space. But since closed subspaces are not closed under unions, this is not a simple subset algebra. In particular, although we can take the meet\footnote{We use `meet' and `join' to refer to the lattice-theoretic operations of greatest lower bound and least upper bound, respectively.} of two subspaces to be the intersection, we need to define the join of two subspaces as 
the closed linear sum, not the union. And as Birkhoff and von Neumann pointed out \cite{BirvNe36}, these two operations do not satisfy the distributive laws, \ie 
\begin{align*}
a \join (b \meet b) &= (a \join b) \meet (a \join c),\\
a \meet (b \join c) &= (a \meet b) \join (a \meet c),
\end{align*}
are not guaranteed to hold where $a, b, c$ are subspaces of a Hilbert space $\cH$ and $\meet, \join$ are the lattice operations given by intersection and closed linear sum, respectively. So, in the standard Hilbert space formulation of QM, the logic governing the physical propositions associated with a quantum system is non-distributive. Given the central role played by the distributive law in classical and intuitionistic logic, this suggests a radically non-classical form of quantum logic. 

The strongest analogue of distributivity that is generally satisfied by the algebra of projections on a Hilbert space is known as the \emph{orthomodular law}. Before stating the law, note that the lattice $\PH$ of projections on a Hilbert space $\cH$ also comes equipped with a canonical orthocomplementation operation $^\bot: \PH \rightarrow \PH$ that takes each subspace to its orthogonal complement. This orthocomplementation operation is fully classical in the sense that it satisfies excluded middle (\ie $a\join a^{\bot}=\top$, where $\top$ is the identity operator, the projection onto the whole of Hilbert space, and hence the top element of $\PH$) and non-contradiction (\ie $a\meet a^{\bot}=\bot$, where $\bot$ is the zero projection onto the null subspace,\footnote{Throughout the paper, we will use $\top$ and $\bot$ to denote the maximal and minimal elements of the relevant lattice/algebra, respectively. We trust that the context will make it clear which algebra $\top$ and $\bot$ inhabit. We also trust that no confusion will arise between (the notation for) the bottom element and the orthocomplementation operation.}
 and hence the bottom element in $\PH$), and is an involution (\ie $a \leq b$ implies $a^{\bot} \geq b^{\bot}$).

\begin{definition}
An orthocomplemented lattice $L$ is called \emph{orthomodular} if and only if for any $a, b \in L$, $a \leq b$ implies that there exists a Boolean sublattice $B \subseteq L$ such that $a, b \in B$.\footnote{Note that this definition is equivalent to the more common, but less instructive algebraic definition, which says that $a \leq b$ implies $a = a \join (b^{\bot} \meet a)$ in an OML.} An orthomodular lattice $L$ is called \emph{complete} if meets and joins of arbitrary families of elements in $L$ exist.
\end{definition}

It is a basic result that the projection operators on a Hilbert space always form a complete orthomodular lattice (complete OML). Thus, quantum logic is standardly characterised as the study of logical systems whose algebraic semantics are given by complete OML's. However, a number of new technical and conceptual difficulties arise once one moves to a non-distributive setting. One issue that will be important for our purposes concerns the definition of a canonical implication operation. It is well known that any distributive lattice uniquely defines a corresponding implication operation. But once we surrender distributivity, the choice of implication operation is no longer obvious. Hardegree \cite{Har81} identifies three 
`quantum material' implication connectives characterized by
the minimal implicative criteria for any orthomodular lattice. 

\begin{definition}	\label{Def_MinImplications}
The \emph{quantum material 
implication connectives} on an orthomodular lattice are the following:

{\centering
\begin{varwidth}{\textwidth}
\begin{enumerate}[{\rm (1)}]
\item (Sasaki conditional) $a \rightarrow_{{\rm S}} b = a^{\bot} \join (a \meet b)$.
\item (Contrapositive Sasaki conditional) $a \rightarrow_{{\rm C}} b = b^{\bot} \rightarrow_{{\rm S}} a^{\bot}$.
\item (Relevance conditional) $a \rightarrow_{{\rm R}} b = (a \rightarrow_{{\rm S}} b) \meet (a \rightarrow_{{\rm C}} b)$. 
\end{enumerate}
\end{varwidth}
\par}
\end{definition}

Hardegree \cite{Har81} proved that these are the only three orthomodular polynomial connectives $\rightarrow$
satisfying {the} 
minimal implicative criteria: (1) modus ponens and modus tollens and 
(2) reflecting and preserving order, \ie  $a \leq b$ if and only if $a \rightarrow b = \top$.
\section{QST: The Basics}	\label{Sec_QST}

We turn now to providing a concise introduction to QST. Firstly, we need to recall some facts about algebraic valued models of set theory. 
\subsection{Boolean Valued Models}

Boolean valued models of set theory were originally developed in the 1960's by Scott, Solovay and Vopenka, as a way of providing a novel perspective on the independence proofs of Paul Cohen.\footnote{For a complete survey of the approach, see Bell \cite{Bel11}.} Intuitively, they can be thought of as generalizations of the usual universe $V$ (the `ground model') of classical set theory (ZFC).

To see this, assume, by recursion, that for any set $x$ in the ground model $V$ of ZFC, there exists a unique corresponding characteristic function $f_{x}$ such that $x \subseteq \dom(f_{x})$ and $\forall y \in \dom(f_{x}) ((f_{x}(y) = 1$ $\leftrightarrow$ $y \in x) \meet (f_{x}(y) = 0 \leftrightarrow y \notin x))$. This assignment allows us to think of $V$ as a class of two-valued functions from itself into the two valued Boolean algebra $\textbf{2} = \{0, 1\}$. This can all be done rigorously with the following definition, by transfinite recursion on $\alpha$:
\[
V_{\alpha}^{(2)} = \{x: \func(x) \meet \ran(x) \subseteq \textbf{2} \meet \exists \xi < \alpha (\dom(x) \subseteq V_{\xi}^{(2)})\}.
\]
Then, we collect each stage of the hierarchy into a single universe by defining the class term 
\[
V^{(2)} = \bigcup_{\alpha} V^{(2)}_{\alpha}.
\]
This allows us to give the following intuitive characterisation of elements of $V^{(2)}$:
\[
x \in V^{(2)}\mbox{ iff }\func(x) \meet \ran(x) \subseteq 2 \meet \dom(x) \subseteq V^{(2)}.
\]

So far, we have simply taken the usual ground model of ZFC and built an equivalent reformalisation of it with no apparent practical purpose. However, we are now in a position to ask another interesting question. What happens if, rather than considering functions into the two-element Boolean algebra, \textbf{2}, we generalise and consider functions whose range is included in an arbitrary but fixed complete Boolean algebra (CBA), $B$? In this case, we obtain the following definition.

\begin{definition}	\label{Def_BooleanValuedModel}
Given a CBA B, the \emph{Boolean valued model} generated by B is the structure $V^{(B)} = \bigcup_{\alpha} V^{(B)}_{\alpha}$, where
\[
V_{\alpha}^{(B)} = \{x: \func(x) \meet \ran(x) \subseteq B \meet \exists \xi < \alpha (\dom(x) \subseteq V_{\xi}^{(B)})\}.
\]
\end{definition}

Now, the natural question to ask here is whether or not these generalised Boolean valued models still provide us with a model of ZFC, and if so, whether the model we obtain in this way has properties that differ from those of the ground model. But we are not yet in a position to provide a sensible answer to that question. For, we do not yet have access to a clear description of what it means for a sentence $\phi$ of the language of set theory
 to hold or fail to hold in $V^{(B)}$ for an arbitrary fixed CBA $B$. So, our immediate task is to define a suitable satisfaction relation for formulae of the language in $V^{(B)}$.\footnote{Of course, we augment the language $L$ of ZFC by adding names for the elements of $V^{(B)}$. We call formulae and sentences in this language `\emph{$B$-formulae}' and `\emph{$B$-sentences}', respectively.} We do this by assigning to each sentence $\phi$ of the language an element $\| \phi \|^{B}$ of $B$, called the `Boolean truth value' of $\phi$.\footnote{We omit the superscript from $\| \phi \|^{B}$ when the CBA $B$ is clear from context.}
We proceed via induction on the complexity of formulae. 

Suppose that Boolean truth values have already been assigned to all atomic $B$-sentences (sentences of the form $u = v, u \in v$, for $u,v \in V^{(B)}$). Then, for $B$-sentences $\phi$ and $\psi$, we set 

\bigskip
{\centering
\begin{varwidth}{\textwidth}
\begin{enumerate}[{\rm (i)}]
\item $\| \phi \meet \psi \| = \| \phi \| \meet \| \psi \|$
\item $\| \phi \join \psi \| = \| \phi \| \join \| \psi \|$
\item $\| \neg \phi\| = \neg \| \phi \|$
\item $\| \phi \rightarrow \psi \| = \| \phi \| \Rightarrow \| \psi \|$
\end{enumerate}
\end{varwidth}
\par}
\bigskip

\noindent where $\neg$ and $\Rightarrow$ represent the complementation and implication operations of $B$. If $\phi(x)$ is a $B$-formula with one free variable $x$ such that $\| \phi(u) \|$ has already been defined for all $u \in V^{(B)}$, we define

\bigskip
{\centering
\begin{varwidth}{\textwidth}
\begin{enumerate}
\item[(v)] $\|\forall x \phi (x) \| =\displaystyle \bmeet_{u \in V^{(B)}}$ $\| \phi (u) \|$
\item[(vi)] $\|\exists x \phi (x) \| = \displaystyle\bjoin_{u \in V^{(B)}}$ $\| \phi (u) \|$

\end{enumerate}
\end{varwidth}
\par}
\bigskip

We also define `bounded' quantifiers $\exists x\in v $ and 
$\forall x\in v$ for every $v\in V^{(B)}$ as

\bigskip
{\centering
\begin{varwidth}{\textwidth}
\begin{enumerate}
\item[(vii)] $\|\forall x\in v\, \phi (x) \| = \displaystyle \bmeet_{u \in \dom(v)} (v(u) \Rightarrow\| \phi (u) \|)$
\item[(viii)] $\|\exists x\in v\, \phi (x) \| = \displaystyle \bjoin_{u \in \dom(v)} (v(u) \meet\| \phi (u) \|)$
\end{enumerate}
\end{varwidth}
\par}
\bigskip
So it only remains to define Boolean truth values for atomic $B$-sentences. For technical reasons, it turns out that the following simultaneous definition is best.

\bigskip
{\centering
\begin{varwidth}{\textwidth}
\begin{enumerate}
\item[(ix)]  $\| u = v \| = \|\forall x\in u\, (x\in v)\| \meet \|\forall y\in v\, (y\in u) \|$
\item[(x)] $\| u \in v \|=\|\exists y\in v\, (u=y)\|$
\end{enumerate}
\end{varwidth}
\par}
\bigskip

We have now defined Boolean truth values for all $B$-sentences. We say that a $B$-sentence $\sigma$ `holds', `is true', or `is satisfied' in $V^{(B)}$ if and only if $\| \sigma \| = \top$. In this case, we write $V^{(B)} \models \sigma$. This is our satisfaction relation. We are now able to ask whether or not $V^{(B)}$ is a model of ZFC. All we have to do is take each axiom and check to see if its Boolean truth value is $\top$. The basic theorem of Boolean valued set theory\footnote{For details and a proof, see chapter 1 of Bell \cite{Bel11}.} is that this is indeed the case. All the axioms of ZFC hold in $V^{(B)}$, regardless of which CBA you choose as your truth value set $B$. However, although all CBA's agree on the truth of the axioms of ZFC, they do not generally agree on the truth of all set theoretic sentences. This is what makes Boolean valued models useful for independence proofs. For example, there are certain choices of $B$ that will make $V^{(B)}$ a model of the continuum hypothesis (CH), and certain choices that will make $V^{(B)}$ a model of $\neg$CH, which demonstrates the independence of CH from the standard axioms of ZFC. 

The following embedding of $V$ into $V^{(B)}$ (for arbitrary but fixed $B$) will play an important role in later sections.
\begin{align*}
\hat{}: V &\rightarrow V^{(B)},\\
x &\mapsto \hat{x}, 
\end{align*}
where $\hat{x} = \{\langle \hat{y}, \top \rangle \: | \: y \in x \} $, \ie $\dom({\hat{x}}) = \{ \hat{y} \: | \: y \in x \}$ and $\hat{x}$ assigns the value $\top$ to every element of its domain. It is easily shown\footnote{See chapter 1 of Bell \cite{Bel11}.} that this embedding satisfies the following properties:

\[
\left\{
\begin{tabular}
[c]{ll}%
$x \in y${ iff }$\| \hat{x} \in \hat{y} \| = \top,$\\
$x \notin y${ iff }$\| \hat{x} \in \hat{y} \| = \bot$.
\end{tabular}
\ \right.
\]

\[
\left\{
\begin{tabular}
[c]{ll}%
$x = y${ iff }$\| \hat{x} = \hat{y} \| = \top,$\\
$x = y${ iff }$\| \hat{x} = \hat{y} \| = \bot$.
\end{tabular}
\ \right.
\]
\smallskip

These properties allow us to establish the following useful results {\cite[Theorem 1.23, Problem 1.24]{Bel11}}:\footnote{Where $\Delta_{0}$-formulae contain only restricted quantifiers and $\Sigma_{1}$-formulae contain no unrestricted universal quantifiers.
{For details, see chapter 13 of Jech \cite{Jec03}.}}
\begin{enumerate}[(1)]
\item
Given a $\De_{0}$-formula $\phi$ with $n$ free variables, and $x_{1},...,x_{n} \in V$, $\phi (x_{1}, ... , x_{n}) \leftrightarrow V^{(B)} \models \phi (\hat{x}_{1}, ... , \hat{x}_{n})$. 
\item
Given a $\Sigma_{1}$-formula $\phi$ with $n$ free variables, and $x_{1},...,x_{n} \in V$, $\phi (x_{1}, ... , x_{n}) \rightarrow V^{(B)} \models \phi (\hat{x}_{1}, ... , \hat{x}_{n})$. 
\end{enumerate}

Since $V^{(B)}$ is a full model of ZFC, it is possible to construct representations of all the usual objects of classical mathematics inside of $V^{(B)}$.  For now, we will be concerned with the real numbers. Specifically, note that the formula defining the predicate $\textbf{R}(x)$ (`x is a Dedekind real number' {or more precisely `x is an upper segment of 
a Dedekind cut of the rational numbers without endpoint'}
) is the following.
\begin{align*}
\mathbf{R}(x) := \forall y \in x (y \in \hat{\mathbb{Q}}) \meet &\exists y \in \hat{\mathbb{Q}} (y \in x) \meet 
\exists y \in \hat{\mathbb{Q}} (y \notin x) \meet\\
&\forall y \in \hat{\mathbb{Q}} (y \in x \leftrightarrow 
\exists z \in \hat{\mathbb{Q}} (z< y \meet z \in x)),
\end{align*}
where $z<y$ is shorthand for the formula $(z,y)\in \hat{R}$, where $R$ is the ordering on $\mathbb{Q}$,
\ie $(z,y)\in R$ if and only if $z\in\mathbb{Q}$, $y\in\mathbb{Q}$, and $z<y$.
Since this is a $\De_{0}$-formula, we know that for any $x \in V$, $V^{(B)} \models \mathbf{R}(\hat{x}) \leftrightarrow x \in \mathbb{R}$. We can use the following definition to construct `the set of all real numbers in $V^{(B)}$'. 

\begin{definition}	\label{Def_R(B)}
Define $\mathbb{R}^{(B)} \subset V^{(B)}$ by 
\[
\mathbb{R}^{(B)} = \{u \in V^{(B)}| \dom(u) = \dom( \hat{\mathbb{Q}}) \meet \|\mathbf{R}(u)\| = \top\}.
\]
\end{definition}

Note that $\mathbb{R}^{(B)}$ is actually not an element of the model $V^{(B)}$. However, we can easily represent $\mathbb{R}^{(B)}$ inside of $V^{(B)}$ in the following way. 

\begin{definition}	\label{Def_R_(B)}
Define $\mathbb{R}_{B} \in V^{(B)}$ by 
\[
\mathbb{R}_{B} = \mathbb{R}^{(B)} \times \{\top\}.
\]
\end{definition}

Think of $\mathbb{R}_{B}$ as the `internal representation' of the real numbers in the model $V^{(B)}$.

\subsection{Boolean Valued Models Generated by Projection Algebras}

Returning to the quantum setting, suppose that we are considering a quantum system represented by a fixed Hilbert space $\cH$ with a corresponding orthomodular lattice of projection operators $\PH$. Then, we can choose a complete Boolean subalgebra $B$ of $\PH$ and build the associated Boolean valued model $V^{(B)}$ in the usual way. The following theorem (due to Takeuti \cite{Tak74}) is the basic founding result of QST. 

\begin{theorem}\label{Thm_InBooleanSubalgRealsAreSelfadjOps} If $B$ is a complete Boolean 
{subalgebra}
of the lattice $\PH$ of projections on a given Hilbert space $\cH$, then the set $\mathbb{R}^{(B)}$ of all real numbers 
in the corresponding model $V^{(B)}$ is isomorphic to the set $SA(B)$ of all self-adjoint operators on $\cH$ 
whose spectral projections all lie in $B$.
\end{theorem}
Intuitively, a Boolean subalgebra $B \subseteq \PH$ corresponds to a set of commuting observables for the relevant quantum system.\footnote{Of course, $B$ can contain eigenprojections for non-commuting operators, but then the idea is that these operators can be simultaneously measured, \emph{to the degree of accuracy determined by the uncertainty relations} \cite{Oza04}.} Commutativity ensures that it is always, in principle, possible to measure these observables at the same time. Thus, any such subalgebra can be thought of as a \emph{classical measurement context} or a \emph{classical perspective} on the quantum system. If we think only about the physical propositions included in $B$, we can reason classically without running into any strange quantum paradoxes. It is only when we try to extend the algebra to include non-commuting projections that things start to go wrong.

Theorem \ref{Thm_InBooleanSubalgRealsAreSelfadjOps} tells us that, to each classical measurement context $B$, there corresponds a set-theoretic universe in which the real numbers are isomorphic to the algebra of physical quantities that we can talk about in that measurement context. This is a deep and enticing result that suggests the possibility of finding mathematical universes whose internal structures are uniquely well suited to describing the physics of specific quantum systems. Towards this end, it is natural to ask whether one can find a model $M$ such that the real numbers in $M$ are isomorphic not just to the set $SA(B)$ of self-adjoint operators whose spectral projections lie in a Boolean subalgebra $B$, but rather to the whole set $SA(\cH)$ of \emph{all} self-adjoint operators on the Hilbert space $\cH$ associated with a given system. Such a model would truly encode the physics of the given system in a complete and robust manner. 

\subsection{Orthomodular Valued Models}
In order to obtain the desired extension of theorem \ref{Thm_InBooleanSubalgRealsAreSelfadjOps}, it is natural to embed the individual Boolean valued models $V^{(B)}$ ($B \subseteq \PH$) within the larger structure $V^{(\PH)}$. This is the strategy introduced by Takeuti \cite{Tak81} and subsequently studied by Ozawa (see \eg Ozawa \cite{Oza07,Oza16b,Oza16,Oza17}), Titani (\cite{Tit99}) and others. But of course, since $\PH$ is a non-distributive lattice, $V^{(\PH)}$ will not be a full model of any well known set theory. However, Ozawa \cite{Oza07} proved a theorem transfer principle that guarantees the satisfaction of significant fragments of classical mathematics within $V^{(\PH)}$. We turn now to outlining this useful result. First, we need some basic definitions. 

{
The structure $V^{(\PH)}$ is defined in a manner parallel with the Boolean-valued model
$V^{(B)}$.  For the later discussions, here we shall give a formal definition for a lattice-valued
model; see Titani \cite{Tit99} for a similar approach.

Let $L$ be a complete lattice.
The $L$-valued model $V^{(L)}$ is the structure $V^{(L)} = \bigcup_{\alpha} V^{(L)}_{\alpha}$, where
\[
V_{\alpha}^{(L)} = \{x: \func(x) \meet \ran(x) \subseteq L \meet \exists \xi < \alpha (\dom(x) \subseteq V_{\xi}^{(L)})\}.
\]
For every $u\in V^{(L)}$, the rank of $u$ 
is defined as the least $\alpha$ such that $u\in V_{\alpha+1}^{(L)}$. 
It is easy to see that if $u\in\dom(v)$ then ${\rm rank}(u) < {\rm rank}(v)$.

\bDefinition \label{intrepretation}
Let $(L,\Rightarrow, *)$ be a triple consisting of 
a complete lattice $L$ with a binary operation 
$\Rightarrow$ (implication)
and a unary operation $*$ (negation).
For any sentence $\phi$ of the language of set theory augmented by the names of elements
of $V^{(L)}$ we define the $L$-valued truth value $\val{\ph}$, called the
$(L,\Rightarrow,*)$-interpretation of $\phi$, 
by the following rules recursive on the rank of elements
of $V^{(L)}$ and the complexity of formulas.
\benum
\item $\val{\ph_1\And\ph_2}=\val{\ph_1}\And\val{\ph_2}$.
\item $\val{\ph_1\Or\ph_2}=\val{\ph_1}\Or\val{\ph_2}$.
\item $\val{\Not\ph}=\val{\ph}^*$.
\item $\val{\ph_1\Then\ph_2}=\val{\ph_1}\Rightarrow\val{\ph_2}$.
\item $\val{\ph_1\Iff\ph_2}=(\val{\ph_1}\THEN\val{\ph_2})\And(\val{\ph_2}\THEN\val{\ph_1})$.
\item $\val{\forall x \ph(x)}=\Inf_{u'\in\VL}\val{\ph(u')}$.
\item $\val{\exists x \ph(x)}=\Sup_{u'\in\VL}\val{\ph(u')}$.
\item $\val{\forall x\in u\, \ph(x)}=\Inf_{u'\in\dom(u)}u(u')\Rightarrow \val{\ph(u')}$.
\item $\val{\exists x\in u\, \ph(x)}=\Sup_{u'\in\dom(u)}u(u')\And \val{\ph(u')}$.
\item $\val{x=y}=\val{\forall x'\in x (x'\in y)}\And
\val{\forall y'\in y (y'\in x)}$.
\item $\val{x\in y}=\val{\exists y'\in y\, (x=y')}$. 
\eenum
We write $\VL\models \ph$ if\/ $\val{\ph}=\top$.
\eDefinition
{Combining (vii)--(x) in the above definition, we obtain
\begin{align*}
{\rm (xii)}&\quad \val{x=y}=\Inf_{x'\in\dom(x)} x(x')\THEN\val{x'\in y}\And 
\Inf_{y'\in\dom(y)}y(y')\THEN\val{y'\in x}.\\
{\rm (xiii)}&\quad \val{x\in y}=\Sup_{y'\in\dom(y)}y(y')\And\val{y'=x}.
\end{align*}}

In the case where $L=\PH$, we shall consider the three implications $\Then_{\rm S}$,
$\Then_{\rm C}$, and $\Then_{\rm R}$, introduced in definition \ref{Def_MinImplications},
and one negation $*=\perp$.

}
\begin{definition}\label{Def_PQCommute} Given $a, b \in \PH$, we say that $a$ and $b$ \emph{commute}, or $a\commutes b$, if $a = (a \meet b) \join (a \meet b^{\bot})$.
\end{definition}

Definition \ref{Def_PQCommute} provides a lattice theoretic generalisation of the notion of commuting operators. 

\begin{definition} Given an arbitrary subset $A \subseteq \PH$, let 
\[
A^{!} = \{a \in \PH\mid a\commutes b \Forall b \in A\}.
\]
\end{definition}

We call $A^{!}$ the \emph{commutant of $A$} since it consists of all the elements of $\PH$ that commute with everything in $A$. 

\begin{definition} Given an arbitrary subset $A \subseteq \PH$, let 
\[
\amalg (A) = \bjoin \{a \in A^{!}\mid  (b_{1} \meet a)\commutes (b_{2} \meet a) \Forall b_{1}, b_{2} \in A\}.
\]
\end{definition}

\begin{definition} Given $u \in V^{(\PH)}$, define the \emph{support of $u$}, $L(u)$, by the following transfinite recursion.
\[
L(u) = \bigcup_{x \in \dom(u)} L(x) \cup \{u(x)\mid x \in \dom(u)\}.
\]

For any $A \subseteq V^{(\PH)}$, we will write $L(A)$ for $\bigcup_{u \in A} L(u)$, and for any $u_{1}, ..., u_{n} \in V^{(\PH)}$, 
we write $L(u_{1}, ..., u_{n})$ for $L(\{u_{1}, ..., u_{n}\})$. 
\end{definition}

\begin{definition} \label{def:commutator}
Given $A \subseteq V^{(\PH)}$, define $\ul{\join} (A)$, the \emph{commutator} of $A$, by 
\[
\ul{\join} (A) = \amalg(L(A)).
\]
\end{definition}

\sloppy
Intuitively, $\ul{\join} (u_{1}, ..., u_{n})$ measures the extent to which the orthomodular-valued sets $u_{1}, ..., u_{n} \in V^{(\PH)}$ `commute' with each other. 
It is a generalisation of the usual operator theoretic notion of a commutator to the set theoretic setting \cite[Theorem 5.4]{Oza07}. 
{If $\ul{\join} (u_{1}, ..., u_{n})=1$, 
elements $u_{1}, ..., u_{n} \in V^{(\PH)}$ perfectly `commute' 
with each other, 
and there exists a complete Boolean subalgebra $B$ of $\PH$ 
such that 
$L(u_{1}, ..., u_{n})\subseteq B$ and 
$u_1,\ldots,u_n\in V^{(B)}$ \cite[Proposition 5.1]{Oza17}.}

Before stating Ozawa's \cite{Oza07,Oza17} `ZFC transfer principle', we need to note an important subtlety that arises when one studies orthomodular valued models. Recall definition \ref{Def_MinImplications}, according to which any orthomodular lattice comes equipped with at least three equally plausible implication connectives. When we were defining the semantics for Boolean valued models, the truth values of conditional sentences were always uniquely determined because of the existence of a single canonical Boolean implication connective. But when we define the set-theoretic semantics for $V^{(\PH)}$, it is always necessary to specify a particular choice of implication. We denote the Sasaki, contrapositive and relevance implications by $\Rightarrow_{j}$ where $j \in \{\rS, \rC, \rR\}$. Similarly, when we want to specify the implication connective used to define the semantics on $V^{(\PH)}$, we write $\|\phi\|_{j}$ to denote the truth values of sentences under that semantics. 
\sloppy
\begin{theorem}\label{Thm_ZFCTransfer} For any $\De_{0}$-formula $\phi(x_{1}, ..., x_{n})$ and any $u_{1}, ..., u_{n} \in V^{(\PH)}$, if $\phi(x_{1}, ..., x_{n})$ is provable in ZFC then 
\[
\ul{\join} (u_{1}, ..., u_{n}) \leq \|\phi({u_{1}, ..., u_{n}})\|_{j}
\] 
for $j = {{\rm S}, {\rm  C}, {\rm R}}.$
\end{theorem}

Theorem \ref{Thm_ZFCTransfer} allows us to recover large fragments of classical mathematics within `commutative regions' of $V^{(\PH)}$ (and also limits the non-classicality of `non-commutative regions'). 
Although $V^{(\PH)}$ is certainly not a full model of ZFC, theorem \ref{Thm_ZFCTransfer} will allow us to model significant parts of classical mathematics within $V^{(\PH)}$. 
{Specifically, if $\ul{\join} (u_{1}, ..., u_{n})=1$, there exists a complete Boolean subalgebra $B$
of $\PH$ such that $u_{1}, ..., u_{n}\in V^{(B)}$ and that $\|\phi({u_{1}, ..., u_{n}})\|_{j}=1$
for any $\De_{0}$-formula $\phi(x_{1}, ..., x_{n})$ provable in ZFC.} 
It is also important to note that the proof of theorem \ref{Thm_ZFCTransfer} is independent of our choice of 
quantum material implication connective on $\PH$ \cite{Oza17}.

Just as we used definitions \ref{Def_R(B)} and \ref{Def_R_(B)} to construct the real numbers for Boolean valued models, it is possible to construct $\mathbb{R}^{(\PH)}$, $\mathbb{R}_{\PH}$ in the same way. Using theorem \ref{Thm_ZFCTransfer}, Ozawa \cite{Oza07,Oza16b,Oza16} proves a number of useful results characterising the behaviour of $\mathbb{R}_{\PH}$. 
{In particular, note here that for any $u,v\in\RPH$ we have $\valj{u=v}=\top$ if and only if $u=v$, \ie $u(\ck{r})=v(\ck{r})$ for all $r\in\Q$.
}

Recall that the purpose of studying the structure $V^{(\PH)}$ was to establish a bijection between $\mathbb{R}^{(\PH)}$ 
and the set $SA(\cH)$ of all self-adjoint operators on $\cH$ \cite{Tak74,Tak81,Oza07,Oza16}. 

\begin{theorem}\label{th:q-reals}
Given a Hilbert space $\cH$, there exists a bijection between $SA(\cH)$ and $\mathbb{R}^{(\PH)}$.
\end{theorem}
{
Recall that, by the spectral theorem for self-adjoint operators, to any self-adjoint operator $X \in SA(\cH)$, 
there corresponds a unique left-continuous family of spectral projections $\{E^{X}_{\la}|\la \in \R\} \subseteq P(\cH)$ 
satisfying the following properties.
\bigskip

{\centering
\begin{varwidth}{\textwidth}
\begin{enumerate}
\item[(i)] $\bigvee \limits_{\la \in \R} E^{X}_{\la} = \top$.
\quad (ii) $\Inf \limits_{\la \in \mathbb{R}} E^{X}_{\la} = \bot$.
\quad (iii) $\bigvee \limits_{\mu \in \mathbb{R}: \mu< \la} E^{X}_{\mu} = E^{X}_{\la}$.
\item[(iv)] $\displaystyle\int_{-\infty}^{+\infty}\la \, d(\xi,E^{X}_\la\psi)=(\xi,X\psi)$ for all $\xi\in\cH$\\
\qquad\qquad and $\psi\in\cH$ 
such that 
 $\displaystyle\int_{-\infty}^{+\infty}\la^2 \, d(\|E^{X}_\la\psi\|^2)<+\infty$.
\end{enumerate}
\end{varwidth}
\par}
\bigskip

Any left-continuous family of projections $\{E_{\la}\}_{\la \in \mathbb{R}} \subseteq P(\cH)$ can be uniquely extended to a countably additive 
projection-valued measure $\tE:\cB(\R)\to \PH$ such that $\tE((-\infty,\la))=E_\la$,
where $\cB(\R)$ denotes the $\sigma$-field of Borel subsets of $\R$,
and conversely any countably additive projection-valued measure can be
obtained in this way.  For the spectral theorem in the projection-valued measure form
we refer the reader to Reed and Simon \cite{RS80}.

In our setting, in which  a real number is defined as an upper segment of a Dedekind cut 
of the rational numbers without endpoint, the bijection between $u\in\RPH$ and $X\in SA(\cH)$
in theorem \ref{th:q-reals} is given by the relations
\benum
\item $u(\ck{r})=E^{X}_r$ \Forall $r\in\Q$,
\item $\displaystyle E^{X}_\la=\Sup_{r\in\Q:r<\la}u(\ck{r})$ \Forall $\la\in\R$,
\eenum
where $\{E^{X}_\la\}_{\la\in\R}$ is the {\em left-continuous} spectral family of $X$.
Note that two left-continuous spectral families are equal if they coincide on $\Q$, \ie $X=Y$ if $E^{X}_r=E^{Y}_r$
for all $r\in\Q$.

Note that  in Takeuti's setting \cite{Tak74,Tak81} (see also \cite{Oza07,Oza16b,Oza16}) a real number is defined as an upper segment of a Dedekind cut 
of the rational numbers that is the complement of a lower segment without endpoint.  
In that case, the bijection between $v\in\RPH$ and $X\in SA(\cH)$ is given by the relations
\benum
\item $v(\ck{r})=\oE^{X}_r$ \Forall $r\in\Q$,
\item $\displaystyle \oE^{X}_\la=\Inf_{r\in\Q:\la<r}v(\ck{r})$ \Forall $\la\in\R$,
\eenum
where $\{\oE^{X}_\la\}_{\la\in\R}$ is the {\em right-continuous} spectral family of $X$,
which satisfies the following properties.

\bigskip
{\centering
\begin{varwidth}{\textwidth}
\begin{enumerate}
\item[(i)] $\Inf \limits_{\la \in \R} (\oE^{X}_{\la})^{\bot} = \bot$.
\quad (ii) $\Inf \limits_{\la \in \mathbb{R}} \oE^{X}_{\la} = \bot$.
\quad (iii) $\Inf \limits_{\mu \in \mathbb{R}: \la<\mu} \oE^{X}_{\mu} = \oE^{X}_{\la}$.
\item[(iv)] $\displaystyle\int_{-\infty}^{+\infty}\la \, d(\xi,\oE^{X}_\la\psi)=(\xi,X\psi)$ for all $\xi\in\cH$\\
\qquad\qquad and $\psi\in\cH$ 
such that 
 $\displaystyle\int_{-\infty}^{+\infty}\la^2 \, d(\|\oE^{X}_\la\psi\|^2)<+\infty$.
\end{enumerate}
\end{varwidth}
\par}
\bigskip
Any right-continuous family $\{\oE_\la\}_{\la\in\R}\subseteq\PH$ of projections
can be uniquely extended to a countably additive 
projection-valued measure $\tE:\cB(\R)\to \PH$ such that $\tE((-\infty,\la])=\oE_\la$
for all $\la\in\R$.  Thus, we have $\oE^{X}_{\la}=E^{X}_{\la}\Or\tE(\{\la\})=\Inf_{r\in\Q:\la<r}E^{X}_{r}$
and $E^{X}_{\la}=\Sup_{r\in\R:r<\la}\oE_{r}$ for any $\la\in\R$.

}

In summary then, given a quantum system with a corresponding Hilbert space $\cH$, we can build the set-theoretic structure $V^{(\PH)}$. We know that distributive regions of $V^{(\PH)}$ behave classically in the sense characterised by theorem \ref{Thm_ZFCTransfer}, and we also know that the real numbers in $V^{(\PH)}$ are in bijective correspondence with the physical quantities associated with the given quantum system.

\section{TQT: Basic Structures and Generalised Stone Duality}	\label{Sec_TQT}

We turn now to providing a basic introduction to the logical aspects of TQT. We focus especially on the version based on orthomodular lattices, first developed by Cannon \cite{Can13} in her MSc thesis and extended by Cannon and D\"oring \cite{CanDoe16}. This version is close in most respects to the standard version based on operators and von Neumann algebras \cite{DoeIsh11}, but the connections with quantum logic and quantum set theory are even more direct in the new version.

The philosophical motivation behind TQT is primarily to provide a `neo-realist' reformulation of quantum theory. Specifically, in TQT one identifies the following two properties as characteristic of any realist theory (for an in depth discussion of the realist interpretation of TQT, see Eva \cite{Eva16}).

\begin{itemize}
	\item[(1)] It is always possible to simultaneously assign truth values to all the physical propositions in a coherent and non-contextual way. 
	\item[(2)] The logic of physical propositions is always distributive. 
\end{itemize}

We have already seen that condition 2 is violated by the standard Hilbert space formalism of QM. Furthermore, the Kochen-Specker theorem tells us that condition 1 cannot be satisfied if the truth values in question form a Boolean algebra and the physical propositions are represented by projection operators. More precisely, the Kochen-Specker theorem shows that there is no non-contextual assignment of classical truth values to all the elements of the projection lattice on a Hilbert space of dimension greater than 2 that preserves the algebraic relationships between the operators. If one defines a `possible world' to be an assignment of classical truth values to all the propositions in the algebra (\eg a `row of the truth table'), then the Kochen-Specker theorem tells us (modulo some important caveats) that there are generally no possible worlds for quantum systems.

Of course, in classical physics, the possible worlds are given by the possible states of the system. Any state assigns to every physical proposition a classical truth value, and any assignment of truth values to all physical propositions uniquely defines a corresponding point in the state space. Thus, we can represent a classical state as a Boolean algebra homomorphism $\ld: B \rightarrow \{0, 1\}$ from the Boolean algebra of physical propositions to the two-element Boolean algebra, $\mathbf{2}=\{0,1\}$. 

One basic aim of TQT is to reformulate the logical structure of QM in a way that satisfies conditions 1 and 2 above. We now sketch the main steps in this enterprise.

\subsection{Stone Duality and Stone Representation}

\emph{Stone duality} goes back to Stone's seminal paper \cite{Sto36}. In modern language, Stone duality is a dual equivalence between the category $\BA$ of Boolean algebras and the category $\Stone$ of Stone spaces, that is, compact, totally disconnected Hausdorff spaces,
\[
				\xymatrix{\BA \ar@<1ex>^-{\Sigma}[rr] \ar@{}|-{\bot}[rr] && \Stone^{\op}. \ar@<1ex>^-{\cl(-)}[ll]}
\]
\begin{itemize}
	\item To each Boolean algebra $B$, the set
	\[
				\Sigma_B:=\{\ld:B\rightarrow \mathbf{2} \mid \ld\text{ is a Boolean algebra homomorphism}\}
	\]
	is assigned. $\Sigma_B$, equipped with the topology generated by the sets $S_a:=\{\ld\in\Sigma_B \mid \ld(a)=1\}$, is called the \textbf{Stone space of $B$}. One can easily check that the sets $S_a$ are clopen, that is, simultaneously closed and open, so the Stone topology has a basis of clopen sets and hence is totally disconnected.
	\item Conversely, to each Stone space $X$, the Boolean algebra $\cl(X)$ of its clopen subsets is assigned. 
\end{itemize}
This is the object level of the two functors above. For the arrows (morphisms), we have:
\begin{itemize}
	\item To each morphism $\phi: B\rightarrow C$ of Boolean algebras, the map
\begin{align*}
			\Sigma(\phi):\Sigma_C &\longrightarrow \Sigma_B\\
			\ld &\longmapsto\ld\circ\phi
\end{align*}
is assigned.
	\item Conversely, to each continuous function $f:X\rightarrow Y$ between Stone spaces, the Boolean algebra morphism
\begin{align*}
			\cl(f):\cl(Y) &\longrightarrow \cl(X)\\
			S &\longmapsto f^{-1}(S)
\end{align*}
is assigned.
\end{itemize}
Note the reversal of direction and `action by pullback' in both cases. Stone duality comes in different variants: there is also a dual equivalence
\[
				\xymatrix{\cBA \ar@<1ex>^-{\Sigma}[rr] \ar@{}|-{\bot}[rr] && \Stonean^{\op}. \ar@<1ex>^-{\cl(-)}[ll]}
\]
between complete Boolean algebras and their morphisms and \emph{Stonean spaces}, which are extremely disconnected compact Hausdorff spaces. The appropriate morphisms between Stonean spaces are open continuous maps.

The \emph{Stone representation theorem} shows that every Boolean algebra $B$ is isomorphic to the Boolean algebra $\cl(\Sigma_B)$ of clopen subsets of its Stone space. The isomorphism is concretely given by
\begin{align*}
			\phi: B &\longrightarrow \cl(\Sigma_{B})\\
			b &\longmapsto S_b:=\{ \ld \in \Sigma_{B} \mid \ld (b) = 1\}.
\end{align*}
While being enormously useful in classical logic, Stone duality and Stone representation do not generalise in any straighforward way to non-distributive orthomodular lattices. In a sense, we can interpret the Kochen-Specker theorem as telling us that the (hypothetical) Stone space of a projection lattice $\PH$ is generally \emph{empty}. Since there are no lattice homomorphisms from $\PH$ into \textbf{2}, there is simply no general notion of the Stone space of an orthomodular lattice. 

This situation was remedied in \cite{Can13,CanDoe16}. In this paragraph, we summarise the main results, which will be presented in more detail in the following subsections. For full proofs and many more details, see \cite{Can13,CanDoe16}. To each orthomodular lattice $L$, one can associate a \emph{spectral presheaf} $\SigL$, which is a straigthtforward generalisation of the Stone space of a Boolean algebra. The assignment is contravariantly functorial, and $\SigL$ is a complete invariant of $L$. Moreover, the clopen subobjects of the spectral presheaf $\SigL$ form a complete bi-Heyting algebra $\Subcl(\SigL)$, generalising the Boolean algebra $\cl(\Sigma_B)$ of clopen subsets of the Stone space of a Boolean algebra $B$. For each complete orthomodular lattice (OML) $L$, there is a map called \emph{daseinisation}, 
\begin{align*}
			\pde: L &\longrightarrow \Subcl(\SigL)\\
			a &\longmapsto \pde(a),
\end{align*}
which provides a representation of $L$ in $\Subcl(\SigL)$. Different from the classical case, $\pde$ is not an isomorphism, but it is injective, monotone, join-preserving, and most importantly, it has an adjoint $\eps:\Subcl(\SigL)\rightarrow L$ such that $\eps\circ\pde$ is the identity on $L$. Hence, the daseinisation map $\pde$ (and its adjoint $\eps$) provide a bridge between the orthomodular world and the topos-based world, at least at the propositional level. In later sections, we will extend this bridge to the level of predicate logic.

\subsection{The Spectral Presheaf}
Given an orthomodular lattice $L$, let $\BL$ be the set of Boolean subalgebras of $L$, partially ordered by inclusion. $\BL$ is called the \emph{context category of $L$}. If $L$ is complete, it makes sense to consider the poset $\BcL$ of complete Boolean subalgebras.

Let $\Lat$ be the category of lattices and lattice morphisms, let $\OML$ denote the category of orthomodular lattices and orthomorphisms (lattice morphisms that preserve the orthocomplement), and let $\Pos$ be the category of posets and monotone maps. When going from an orthomodular lattice $L$ to its context category $\BL$, seemingly a lot of information is lost. By only considering Boolean subalgebras, any non-distributivity in $L$ is discarded. Moreover, since $\BL$ is just a \emph{poset}, the inner structure of each $B\in\BL$ as a Boolean algebra is lost (or so it seems). Somewhat surprisingly, Harding and Navara \cite{HarNav11} (see also \cite{HHLN19}) proved that $\BL$ is a complete invariant of $L$ :
\begin{theorem}
(Harding, Navara) Let $L,M$ be orthomodular lattices with no $4$-element blocks (\ie no maximal Boolean subalgebras with only $4$ elements). Then $L\simeq M$ in $\OML$ if and only if $\BL\simeq\BM$ in $\Pos$.
\end{theorem}
We also note that every morphism $\phi:L\rightarrow M$ between orthomodular lattices induces a monotone map
\begin{align*}
			\tphi: \BL &\longrightarrow \BM\\
			B &\longmapsto \phi[B].
\end{align*}
This means that we have a (covariant) functor from $\OML$ to $\Pos$, $L\rightarrow\BL$, $\phi\rightarrow\tphi$. There are obvious analogues of these constructions for complete OMLs.

While $\BL$ certainly cannot be regarded as a generalised Stone space of $L$, it is significant that $\BL$ is a complete invariant and that the assignment $L\mapsto\BL$ is functorial. This suggests to consider structures built over $\BL$ that are related to Stone spaces. In fact, the simplest idea works: each $B\in\BL$ is a Boolean algebra, so it has a Stone space $\Sigma_B$. If $B'\subset B$, then there is a canonical map from $\Sigma_B$ to $\Sigma_{B'}$, given by restriction:
\begin{align*}
			r(B'\subset B): \Sigma_B &\longrightarrow \Sigma_{B'}\\
			\ld &\longmapsto \ld|_{B'}.
\end{align*}
It is well-known that this map is surjective and continuous with respect to the Stone topologies. Moreover, being a quotient map between compact Hausdorff spaces, it is a closed map. If $B$ and $B'$ are complete, then $\Sigma_B$ and $\Sigma_{B'}$ are Stonean and $r(B'\subset B)$ is also open.

This allows us to define a presheaf over $\BL$ that comprises all the `local' Stone spaces $\Sigma_B$ and glues them together:
\begin{definition}
The \emph{spectral presheaf}, $\SigL$, of an orthomodular lattice $L$ is the presheaf over $\BL$ given
\begin{itemize}
	\item [(a)] on 
	objects: for all $B\in\BL$, $\SigL_{B} := \Sigma_{B}$ (where $\SigL_B$ denotes the component of $\SigL$ at $B$),
	
	\item [(b)] on 
	arrows: for all $B,B'\in\BL$ such that $B' \subset B$, $\SigL(B'\subset B):=r(B'\subset B)$.
\end{itemize}
\end{definition}
If $L$ is a complete OML, then there is an obvious analogue, the spectral presheaf of $L$ over $\BcL$. Obviously, the spectral presheaf $\SigL$ of an orthomodular lattice $L$ is a generalisation of the Stone space $\Sigma_B$ of a Boolean algebra $B$.\footnote{For now we abide by the convention used in the literature on TQT of \underline{underlining} presheaves. We will drop this convention in section \ref{Sec_ParaconsAndDistribInQL} to keep things neat.}

In order to discuss whether the assignment $L\mapsto\SigL$ is functorial, we first have to define a suitable category in which the spectral presheaf $\SigL$ is an object, and the appropriate morphisms between such presheaves. This is done in great detail in \cite{Can13,CanDoe16}. One considers a category $\Presh{\Stone}$ of presheaves with values in Stone spaces, but over varying base categories. (We saw that if $L$ is not isomorphic to $M$, then their context categories $\BL,\BM$ are not isomorphic, either.)\footnote{Alternatively, one could consider a category of presheaf topoi over different base categories, together with a distinguished spectral object in each topos.} With this in place, Cannon \cite{Can13} showed:
\begin{theorem}
There is a \emph{contravariant} functor
\begin{align*}
			\Sig: \OML &\longrightarrow \Presh{\Stone}
\end{align*}
sending an object $L$ of $\OML$ to $\SigL$, which is an object in $\Presh{\Stone}$, and a morphism $\phi:L\rightarrow M$ in $\OML$ to a morphism $\pair{\tphi}{\mathcal{G}_\phi}:\SigM\rightarrow\SigL$ in $\Presh{\Stone}$ in the opposite direction.

If $L,M$ are orthomodular lattices, then $L\simeq M$ in $\Lat$ if and only if $\SigL\simeq\SigM$ in $\Presh{\Stone}$.
\end{theorem}
Again, there is an obvious version of this result for complete OMLs. This shows that the spectral presheaf $\SigL$ is a complete invariant of an orthomodular lattice $L$, and the assignment $L\rightarrow\SigL$ is contravariantly functorial. This strengthens the interpretation of $\SigL$ as a generalised Stone space. Yet, there are two caveats: the behaviour of the orthocomplement under the assignment $L\mapsto\SigL(L)$ was not considered in \cite{Can13} and only very briefly in \cite{CanDoe16}; this will be done in section \ref{Sec_ParaconsAndDistribInQL} below. More importantly, the result is weaker than a full duality, since not every object in $\Presh{\Stone}$ is in the image of the functor $\Sig$.

\subsection{The Algebra of Clopen Subobjects}
Before we can generalise the Stone representation, we define $\Subcl(\SigL)$, the collection of clopen subobjects of the spectral presheaf, and briefly consider its algebraic structure. The guiding idea is that clopen subobjects of the spectral presheaf of an OML generalise the clopen subsets of the Stone space of a Boolean algebra.

\begin{definition}
Let $L$ be an orthomodular lattice, $\BL$ its context category, and $\SigL$ its spectral presheaf. A subobject (\ie a subpresheaf) $\ps S$ of $\SigL$ is called \emph{clopen} if for all $B\in\BL$, the component $\ps S_B$ of $\ps S$ at $B$ is a clopen subset of $\SigL_B=\Sigma_B$, the Stone space of $B$. The collection of clopen subobjects is denoted $\Subcl(\SigL)$.
\end{definition}

We first observe that $\Subcl{\SigL}$ is a partially ordered set with a natural order given by
\begin{align*}
			\forall \ps S,\ps T\in\Subcl{\SigL}: \ps S\leq\ps T :\Longleftrightarrow (\forall B\in\BL:\ps S_B\subseteq\ps T_B).
\end{align*}
Clearly, the empty subobject $\ps\emptyset$ is the bottom element and $\SigL$ is the top element. Moreover, meets and joins exist with respect to this order and are given as follows: for every finite family $(\ps S_i)_{i\in I}\subset\Subcl(\SigL)$ and for all $B\in\BL$,
\begin{align*}
			(\bmeet_{i\in I}\ps S_i)_B &= \bmeet_{i\in I}\ps S_{i;B},\\
			(\bjoin_{i\in I}\ps S_i)_B &= \bjoin_{i\in I}\ps S_{i;B},
\end{align*}
where $\ps S_{i;B}$ denotes the component of $\ps S_i$ at $B$. If $L$ is a complete OML, then we consider the context category $\BcL$ of complete Boolean sublattices and the spectral presheaf $\SigL$ is constructed over $\BcL$. Then \emph{all} meets and joins exist in $\Subcl(\SigL)$, not just finite ones, and they are given as follows: for any family $(\ps S_i)_{i\in I}\subseteq\Subcl\SigL$,
\begin{align*}
			(\bmeet_{i\in I}\ps S_i)_B &= \operatorname{int}(\bmeet_{i\in I}\ps S_{i;B}),\\
			(\bjoin_{i\in I}\ps S_i)_B &= \operatorname{cls}(\bjoin_{i\in I}\ps S_{i;B}),
\end{align*}
where $\operatorname{int}$ denotes the interior and $\operatorname{cls}$ the closure. For now, we go back to general OMLs.

Since meets and joins are calculated componentwise, $\Subcl(\SigL)$ is a distributive lattice. If $L$ is complete, finite meets distribute over arbitrary joins and vice versa. Additionally, it is easy to show that $\Subcl(\SigL)$ is a both a \emph{Heyting algebra} and a \emph{co-Heyting algebra} (or \emph{Brouwer algebra}), hence a \emph{bi-Heyting algebra}. If $L$ is complete, then $\Subcl(\SigL)$ is a complete bi-Heyting algebra \cite{Doe12}. The Heyting algebra structure gives an intuitionistic propositional calculus, the co-Heyting algebra structure a paraconsistent one. The Heyting negation and the co-Heyting negation on $\Subcl(\SigL)$ are not related closely to the orthocomplement in $L$, so we will not consider them further. Instead, we will define a third negation on $\Subcl(\SigL)$ in section \ref{Sec_ParaconsAndDistribInQL} that indeed is closely related to the orthocomplement in $L$ and will prove very useful.

\subsection{Representing a Complete Orthomodular Lattice in the Algebra of Clopen Subobjects by Daseinisation}
From now on, we specialise to complete orthomodular lattices,
{and accordingly we take the spectral presheaf $\SigL$ to be constructed over
the context category $\BcL$ of complete Boolean sublattices.}
 We aim to generalise the Stone representation by defining a suitable map from a complete OML $L$ to the algebra $\Subcl(\SigL)$ of clopen subobjects of its spectral presheaf.

Let $a\in L$, and let $B\in\BcL$ be a (complete) context. If $a\in B$, then we simply use the standard Stone representation and assign the clopen subset $S_a=\{\ld\in\Sigma_B \mid \ld(a)=1\}$ of $\SigL_B=\Sigma_B$ to it. Yet, if $a\notin B$, then we have to approximate $a$ in $L$ first: define
\begin{align*}
			\de_B(a) := \bmeet\{b\in B \mid b\geq a\}.
\end{align*}
The meet exists in $B$, since $B$ is complete. $\de_B(a)$ is the smallest element of $B$ that dominates $a$.
By the Stone representation, there is a clopen subset of $\Sigma_B$ corresponding to $\de_B(a)$, given by $S_{\de_B(a)}=\{\ld\in\Sigma_B \mid \ld(\de_B(a))=1\}$. In this way, we obtain one clopen subset in each component $\Sig_B$ of the spectral presheaf, where $B$ varies over $\BcL$. It is straightforward to check that these clopen subsets form a clopen subobject (see \cite{Can13,CanDoe16}), which we denote by $\pde(a)$,
\begin{align*}
			\forall B\in\BcL: \pde(a)_B = S_{\de_B(a)}.
\end{align*}
We call $\pde(a)\in\Subcl(\SigL)$ the \emph{daseinisation of $a$}.\footnote{{
Note that this is also called \emph{outer daseinisation of $a$} to distinguish it from an inner daseinisation of $a$, which analogously 
approximate $a$ from below.  In this paper, we confine our attention to outer daseinisations, and yet we will extend our arguments
to inner daseinisations elsewhere. }} 

\begin{proposition}\label{th:delta}
Let $L$ be a complete orthomodular lattice, $\BcL$ its complete context category, $\SigL$ its spectral presheaf, and $\Subcl(\SigL)$ the complete distributive lattice of clopen subobjects of $\SigL$. The \emph{daseinisation map}
\begin{align*}
			\pde: L &\longrightarrow \Subcl(\SigL)\\
			a &\longmapsto \pde(a)
\end{align*}
has the following properties:
\begin{enumerate}[{\rm (i)}]
	\item $\pde$ is injective, but not surjective,
	\item $\pde$ is monotone,
	\item $\pde$ preserves bottom and top elements,
	\item $\pde$ preserves all joins, \ie for any family $(a_i)_{i\in I}\subseteq L$, it holds that
		\begin{align*}
					\bjoin_{i\in I} \pde(a_i) = \pde(\bjoin_{i\in I} a_i).
		\end{align*}
	\item[(v)] for meets, we have $\pde(\bmeet \limits_{i \in I} a_{i}) \leq \bmeet \limits_{i \in I}\pde(a_{i})$.
\end{enumerate}
\end{proposition}

\begin{proof}
(i)--(iii) are easy to prove. (iv) follows since meets and joins are calculated componentwise, and $\de_B:L\rightarrow B$ is the left adjoint of the inclusion $B\hookrightarrow L$, so $\de_B$ preserves colimits, which are joins. (v) is a consequence of (ii).
\end{proof}

The map $\pde:L\rightarrow\Subcl(\SigL)$ is our generalisation of the Stone representation $B\rightarrow\cl(\Sigma_B)$. Different from the Stone representation, $\pde$ is not an isomorphism, and it cannot be, since $L$ is non-distributive in general, while $\Subcl(\SigL)$ is distributive. Yet, since $\pde$ is injective, no information is lost. 

Moreover, $\pde$ is a join-preserving map between complete lattices, so it has a right adjoint $\eps:\Subcl(\SigL)\rightarrow L$, given by
\begin{align*}
			\forall \ps S\in\Subcl(\SigL): \eps(\ps S) = \bjoin\{a\in L \mid \pde(a)\leq\ps S\}
\end{align*}
The properties of $\eps$ are analogous to those of $\pde$:
\begin{proposition}
Let $L$ be a complete orthomodular lattice, $\BcL$ its complete context category, $\SigL$ its spectral presheaf, and $\Subcl(\SigL)$ the complete distributive lattice of clopen subobjects of $\SigL$. The right adjoint $\eps$ of $\pde$ has the following properties:
\begin{enumerate}[{\rm (i)}]
	\item $\eps$ is surjective, but not injective,
	\item $\eps$ is monotone,
	\item $\eps$ preserves bottom and top elements,
	\item $\eps$ preserves all meets, \ie for any family $(\ps S_i)_{i\in I}\subseteq\Subcl(\SigL)$, it holds that
		\begin{align*}
					\bmeet_{i\in I} \eps(\ps S_i) = \eps(\bmeet_{i\in I} \ps S_i),
		\end{align*}
	\item[(v)] for joins, we have $\eps(\ps S\join\ps T)\geq\eps(\ps S)\join\eps(\ps T)$.
\end{enumerate}
\end{proposition}
If $B\in\BcL$ is a context and $S\in\cl(\Sigma_B)$, then we denote the element of $B$ corresponding to $S$ under the Stone representation by $a_S$. Carmen Constantin first proved the following useful lemma (see \cite{Can13}):
\begin{lemma}
For all $\ps S\in\Subcl(\SigL)$,
\begin{align*}
			\eps(\ps S) = \bmeet_{B\in\BcL} a_{\ps S_B}.
\end{align*}
\end{lemma}
This means that in order to calculate $\eps(\ps S)$, we simply switch in each context from the clopen set $\ps S_B$ to the corresponding element $a_{\ps S_B}$ of $B$ and then take the meet (in $L$) over all contexts.
\begin{corollary}
With the notation above, {we have $\eps\circ\pde=\id$ on $L$ and $\pde\circ\ep\le\id$ on $\Subcl(\Sig)$}
for any complete orthomodular lattice $L$.
\end{corollary}

$\eps$ can be used to define an equivalence relation on $\Subcl(\Sig)$, given by $\ps S \sim \ps T$ if and only if $\eps(\ps S) = \eps(\ps T)$. We let $E$ denote the set of all equivalence classes of $\Subcl(\Sig)$ under this equivalence relation. $E$ can be turned into a complete lattice by defining $\bmeet_{i \in I} [\ps S_{i}] = [\bmeet_{i \in I} \ps S_{i}]$, $[\ps S] \leq [\ps T]$ if and only if $[\ps S] \meet [\ps T] = [\ps S]$ and $\bjoin_{i \in I} [\ps S_{i}] = \bmeet\{[\ps T] \mid \forall i \in I: [\ps S_{i}] \leq [\ps T]\}$. Then it is straightforward to show

\begin{theorem}	\label{Thm_EAndLIsomorphicAsComplLattices}
With the notation above, $E$ and $L$ are isomorphic as complete lattices. In particular, the maps $g: E \rightarrow L$ and $f: L \rightarrow E$ defined by $g([\ps S]) = \eps (\ps S)$ and $f(a) = [\pde(a)]$ are an inverse pair of complete lattice isomorphisms. 
\end{theorem}

This result can be seen as an alternative generalisation of the Stone representation theorem to orthomodular lattices. It tells us that an orthomodular lattice can be represented isomorphically by the clopen subobjects of the spectral presheaf \emph{modulo the equivalence relation induced by $\eps$}. Of course, the lattice $E=\Subcl(\SigL)/\simeq$ is non-distributive, while $\Subcl(\SigL)$ is distributive. In the following, we will mostly consider $\Subcl(\SigL)$ as the algebra of propositions suggested by the topos approach.

So far, we have seen that TQT allows for an alternative formalisation of the logical structure of QM that does not require us to surrender distributivity. Furthermore, we have also seen that the spectral presheaf allows for the derivation of an analogue of Stone's theorem for the orthomodular setting. However, we have not yet addressed the problem of how we can assign truth values to the newly formalised physical propositions in a coherent way. 

\subsection{Presheaf Topoi over Posets}
We briefly recall some basic definitions from topos theory \cite{McLMoe92,Bel08,Joh0203}, in particular that of a subobject classifier. We then specialise to presheaf topoi and further to presheaf topoi over posets and consider the subobject classifier and truth values in such a topos.

In the following, let $\One$ denote the terminal object in a category (if it exists).
 
\begin{definition}
In a category $\cC$ with finite limits, a \emph{subobject classifier} is a monic, $\true:\One\rightarrow\Omega$, such that to every monic $S\rightarrow X$ in $\cC$ there is a unique arrow $\chi$ which, with the given monic, forms a pullback square
\begin{diagram}
	S & \rTo^! & \One \\
	\dInto^{} & & \dInto_{\true} \\
	X & \rDashto_{\chi} & \Omega. \\
\end{diagram}
\end{definition}
This means that in a category with subobject classifier, every monic is the pullback of the special monic $\true$. In a slight abuse of language, we will often call the object $\Omega$ the subobject classifier. In $\CSet$, we have $\Omega=\{0,1\}$ 
{and $\chi$ is the characteristic function of the set $S$.}

\begin{definition}
An \emph{elementary topos} is a category $\cE$ with the following properties:
\begin{itemize}
	\item $\cE$ has all finite limits and colimits;
	\item $\cE$ has exponentials;
	\item $\cE$ has a subobject classifier.
\end{itemize}
\end{definition}
The structural similarity with the category $\CSet$ of sets and functions -- which is itself a topos, of course -- should be obvious. Some further examples of topoi are:
\begin{itemize}
	\item[(1)] $\CSet\times\CSet$, the category of all pairs of sets, with morphisms pairs of functions;
	\item[(2)] $\CSet^{\mathbf{2}}$, where $\mathbf{2}=\bullet\longrightarrow\bullet$. Objects are all functions from one set $X$ to another set $Y$, with commutative squares as arrows;
	\item[(3)] $\mathbf{B}G$, the category of $G$-sets: objects are sets with a right (or left) action by $G$, and arrows are $G$-equivariant functions;
	\item[(4)] $\CSet^{\cC^{op}}$, where $\cC$ is a small category. This functor category has functors $\cC^{op}\rightarrow\CSet$ as objects (also called \emph{presheaves over $\cC$}) and natural transformations between them as arrows. In fact, all the examples so far are functor categories.
	\item[(5)] The category $Sh(X)$ of sheaves over a topological space $X$.
\end{itemize}
We will now specialise to presheaf topoi, \ie topoi of the form $\CSet^{\cC^{op}}$, where $\cC$ is a small category. The terminal object $\One$ in $\CSet^{\cC^{op}}$ is given by the presheaf that assigns the one-element set $\One_C=\{*\}$ to every object $C$ in $\cC$ and the constant function $\{*\}\rightarrow\{*\}$ to every morphism in $\cC$.

It is a standard result (see \eg \cite{McLMoe92}) that the subobject classifier $\Omega$ in a presheaf topos is given by the \emph{presheaf of sieves}, which we now define. First, let $C$ be an object in $\cC$. A \emph{sieve on $C$} is a collection $\sigma$ of arrows with codomain $C$ such that if $(f:B\rightarrow C)\in\sigma$ and $g:A\rightarrow B$ is any other arrow in $\cC$, then $f\circ g:A\rightarrow C$ is in $\sigma$, too.

If $h:C\rightarrow D$ is an arrow in $\cC$ and $\sigma$ is a sieve on $D$, then $\sigma\cdot h=\{f \mid h\circ f\in \sigma\}$ is a sieve on $C$, the \emph{pullback of $\sigma$ along $h$}. 

\begin{definition}
The presheaf of sieves is given
\begin{itemize}
	\item [(a)] on objects: for all $C\in Ob(\cC)$, $\Omega(C)=\{\sigma \mid \sigma \text{ sieve on }C\}$;
	\item [(b)] on arrows: for all $(h:C\rightarrow D)\in Arr(\cC)$, the mapping $\Omega(h):\Omega(D)\rightarrow\Omega(C)$ is given by the pullback along $h$.
\end{itemize}
\end{definition}

Clearly, this is an object in the topos $\CSet^{\cC^{op}}$. The arrow $\true:\One\rightarrow\Omega$ is given by
\begin{align*}
			\forall C\in Ob(\cC): \true_C: \One_C &\longrightarrow \Omega_C\\
			* &\longmapsto \sigma_m(C),
\end{align*}
where $\sigma_m(C)$ denotes the \emph{maximal sieve} on $C$, \ie the collection of all morphisms in $\cC$ with codomain $C$.

Each topos, and in particular each presheaf topos, comes equipped with an internal, higher-order intuitionistic logic. This logic is multi-valued in general, and the truth values form a partially ordered set. In fact, the truth values in a topos are given by the \emph{global elements} of the subobject classifier $\Omega$, \ie by morphisms
\[
		\One \longrightarrow \Omega
\]
from the terminal object $\One$ to the subobject classifier $\Omega$. Clearly, the arrow $\true:\One\rightarrow\Omega$ represents one such truth value, and it is interpreted as `totally true'. In a presheaf topos $\CSet^{\cC^{op}}$, there also is an arrow $\mathsf{false}:\One\rightarrow\Omega$, which assigns the empty sieve to each $\One_C=\{*\}$, $C\in Ob(\cC)$. The arrow $\mathsf{false}$ represents the truth value `totally false'. In general, there exist other morphisms from $\One$ to $\Omega$, representing truth values between `totally true' and `totally false'.

We now specialise further to presheaf topoi for which the base category $\cC$ is a (small) poset.\footnote{Regarding the poset $\cC$ as a category, there is an arrow $B\rightarrow C$ if and only if $B\leq C$. Hence, in a poset there is at most one arrow from any object $B$ to any object $C$.}

Recall that a sieve on an object $C$ is a collection of morphisms with codomain $C$ that is `downward closed' under composition. In a poset, an arrow $B\rightarrow C$ with codomain $C$ means that $B\leq C$, and if $A\rightarrow B$ is another arrow in the poset category $\cC$ (\ie if $A\leq B$), then the composite arrow $A\rightarrow C$ is also contained in the sieve. Keeping the codomain $C$ fixed, an arrow in the poset category $\cC$ of the form $\tilde B\rightarrow C$ can be identified with its domain $\tilde B$. This means that in a poset $\cC$, a sieve $\sigma$ on $C$ can be identified with a \emph{lower set} (or \emph{downward closed set}) in $\downarrow C=\{\tilde B\in\cC \mid \tilde B\leq C\}$, the downset of $C$. The maximal sieve on $C$ is simply $\downarrow C$.

If $C\rightarrow D$ is an arrow in $\cC$ (\ie if $C\leq D$) and $\sigma$ is a sieve on $D$, then the pullback of $\sigma$ along the arrow $C\rightarrow D$ is simply the lower set in $\downarrow C$ given by $\sigma\;\cap\downarrow C$, as can be checked by inserting into the definition.

We saw that the truth values in a presheaf topos $\CSet^{\cC^{op}}$ are given by the arrows of the form $v:\One\rightarrow\Omega$. Since $\cC$ now is a poset by assumption, the sieve $v_C$ is a lower set in $\downarrow C$, and whenever $C\leq D$, then $v_C=v_D\;\cap\downarrow C$. This means that, overall, the arrow $v:\One\rightarrow\Omega$ defines a lower set in $\cC$. We have shown:
\begin{lemma}
In a presheaf topos $\CSet^{\cC^{op}}$ over a poset $\cC$, the truth values are given by lower sets in the poset $\cC$.
\end{lemma}

\subsection{The Topos $\CSet^{\BcPH^{op}}$, States and Truth Values}
For simplicity, and in order to make closer contact with the usual Hilbert space formalism again, we assume $L=\PH$ in this subsection, where $\PH$ is the projection lattice on a Hilbert space $\cH$ of dimension $3$ or greater. We now show how to employ the topos of presheaves over $\BcPH$, the poset of complete Boolean subalgebras of $\PH$, and its internal logic to assign truth values to propositions.

Let $\ket\psi\in\cH$ be a unit vector, representing a pure state of $\BH$. We apply daseinisation to the rank-$1$ projection $P_\psi=\ket\psi\bra\psi$ and obtain a clopen subobject of $\SigPH$, the spectral presheaf of the complete OML $\PH$. This subobject is denoted as
\begin{align*}
			\wpsi:=\pde(P_\psi)
\end{align*}
and is called the \emph{pseudostate} corresponding to $\ket\psi$. While the spectral presheaf $\SigPH$ has no global sections, a fact that is equivalent to the Kochen-Specker theorem \cite{IshBut98,DoeIsh11}, $\SigPH$ still has plenty of subobjects, and just like a point in the state space of a classical system represents a (pure) state, here the pseudostate $\wpsi$ represents the pure quantum state.

In classical physics, truth values of propositions arise in a simple manner: let $S\subset\mathcal{S}$ be a (Borel) subset of the state space $\mathcal{S}$ of the classical system. $S$ represents some proposition, \eg if $f_A:\mathcal{S}\rightarrow\mathbb{R}$ is the Borel function representing a physical quantity $A$ of the system, and $\Delta\subset\mathbb{R}$ is a Borel subset of the real line, then $S=f_A^{-1}(\Delta)$ is the representative of the proposition ``the physical quantity $A$ has a value in the Borel set $\Delta$''. A (pure) state is represented by a point $p\in\mathcal{S}$. The truth value of the proposition represented by $S$ in the state represented by $p$ is
\begin{align*}
			\|S\|_{p} &= true\text{ if }p\in S\text{ and }\\
			\|S\|_{p} &= false\text{ if }p\notin S.
\end{align*}
Note that this is equivalent to
\begin{align*}
			\|S\|_{p}=(p\in S)
\end{align*}
if we read the right-hand side as a set-theoretic proposition that is either true or false.
Analogously, let $\ps S\in\Subcl(\SigL)$ be a clopen subobject that represents a proposition about a quantum system. Then the truth value of the proposition represented by $\ps S$ in the state represented by $\wpsi$ is
\begin{align*}
			\|\ps S\|_{\wpsi} = (\wpsi\subseteq\ps S).
\end{align*}
Here, the right-hand side must be interpreted as a `set-theoretic' proposition in the topos $\CSet^{\BcPH^{op}}$ of presheaves over $\BcPH$. This is done using the Mitchell-Benabou language of the topos. In our case, where  the base category $\BcPH$ is simply a poset, and $\CSet^{\BcPH^{op}}$ is a presheaf topos, this boils down to something straightforward: for every context $B\in\BcPH$, we consider the set-theoretic expression
\begin{align*}
			\wpsi_B\subseteq\ps S_B
\end{align*}
`locally' at $B$, which can be either $true$ or $false$. We collect all those $B\in\BcPH$ for which the expression $\wpsi_B\subseteq\ps S_B$ is $true$. It is easy to see that in this way we obtain a lower set in $\BcPH$. As we saw in the previous subsection, such a lower set in the poset $\BcPH$, which is the base category of the presheaf topos $\CSet^{\BcPH^{op}}$, can be identified with a global section of the presheaf of sieves on $\BcPH$, and hence with a truth value in the (multi-valued logic provided by the) topos $\CSet^{\BcPH^{op}}$. Hence, we showed how to determine the truth value $\|\ps S\|_{\wpsi}$ of the proposition represented by $\ps S$ in the state represented by $\wpsi$.

There is no obstacle to assigning these truth values in a coherent, non-contextual and structure-preserving way. The Kochen-Specker theorem is not violated, of course: it is a result that applies specifically to the representation of physical propositions as projection operators and truth values as elements of the two-element Boolean algebra. The KS theorem simply does not apply to the reformulated propositions and truth values of TQT, and so it is perfectly possible to simultaneously assign truth values to all physical propositions in TQT.\footnote{For further discussion of the philosophical interpretation of the truth values in TQT, see Eva \cite{Eva16}.}

At this stage, we have sketched the basic formal and conceptual ideas behind QST and TQT. Both approaches offer new Q-worlds in which we can hope to reformulate QM in a conceptually illuminating way. However, beyond this superficial analogy, it is not obvious that there is any deep connection between the two projects. In the next section, we will explore the relationship between the distributive logic of TQT and traditional orthomodular quantum logic. This will subsequently allow us to establish rich and interesting connections between TQT and QST in section 6. 

\section{Paraconsistency and Distributivity in Quantum Logic}   \label{Sec_ParaconsAndDistribInQL}

\subsection{Paraconsistent Negation}
It is interesting to note that theorem \ref{Thm_EAndLIsomorphicAsComplLattices} establishes that $E$ and $L$ are isomorphic \emph{as complete lattices}, but says nothing about the negation operations defined on $E$ and $L$. In section 6, it will be important for our purposes that the isomorphism between $E$ and $L$ preserves negations as well as the lattice structure. There is an obvious way of defining a negation operation on $E$ that allows us to extend the isomorphism from theorem \ref{Thm_EAndLIsomorphicAsComplLattices} to include negations. Specifically, Eva \cite{Eva15} suggests the following definition.

\begin{definition} Given $\ps S \in \Subcl(\Sig)$, define $\ps S^{*} = \pde(\eps(\ps S)^{\bot})$, \ie $\ps S^{*}$ is the daseinisation of the orthocomplement of $\eps(\ps S)$ (where $\bot$ denotes the orthocomplement of $L$).
\end{definition}

The idea is that $\ps S^{*}$ is obtained by translating $\ps S$ into an element of $L$ via $\eps$, negating that element by $L$'s classical orthocomplementation operation, and then translating the negated element back into a clopen subobject via $\pde$. Eva \cite{Eva15} notes that defining $[\ps S]^{*} = [\ps S^{*}]$ implies that $E$ and $L$ are isomorphic not just as complete lattices, but also as complete \emph{ortholattices}. So we can extend the correspondence to cover the full logical structure of $E$. This fact will turn out to be important in our attempts to obtain a connection between TQT and QST. 

However, the negation operation $*$ that we used to extend this isomorphism turns out to have some unexpected properties. Most importantly, $*$ is \emph{paraconsistent}. To see this, recall that for any $\ps S$, $\pde(\eps(\ps S)) \leq \ps S$.
This means that

\[
\ps S \meet \ps S^{*} = \ps S \meet \pde(\eps(\ps S)^{\bot}) \geq \pde(\eps(\ps S)) \meet \pde(\eps(\ps S)^{\bot}) \geq \pde(\eps(\ps S) \meet \eps(\ps S)^{\bot}) = \bot.
\]

These inequalities can all be strict, so $S \meet S^{*}$ will not generally be minimal, \ie the $*$ operation is paraconsistent. So the natural `translation' of the orthocomplement of an orthomodular lattice into the distributive setting is inherently paraconsistent. Eva \cite{Eva15} establishes the following basic properties of the $*$ negation.

\begin{theorem}
The $*$ operation has the following properties:

\bigskip
{\centering
\begin{varwidth}{\textwidth}
\begin{enumerate}[{\rm (i)}]
\item $\ps S \join \ul{S^{*}} = \top$,
\item $\ps S^{**} = \delta(\eps(\ul{S}))\leq \ps S$,
\item $\ps S^{***} = \ps S^{*}$,
\item $\ps S \meet \ps S^{*} \geq \bot$,
\item $\bigvee_j \ps S_j^*=(\bigwedge_j \ps S_j)^*$ for any family $\{\ps S_j\} \subseteq \Subcl(\Sig)$,
\item $\bigwedge_j \ps S_j^*\geq (\bigvee_j \ps S_j)^*$ for any family $\{\ps S_j\} \subseteq \Subcl(\Sig)$,
\item $\eps(\ps S) \join \eps(\ps S^{*}) = \top$,
\item $\eps(\ps S) \meet \eps(\ps S^{*}) =\bot$,
\item $\ps S \leq \ps T$ implies $\ps S^{*} \geq \ps T^{*}$, \ie $*$ is an involution.
\end{enumerate}
\end{varwidth}
\par}
\bigskip

\end{theorem}

Note that before we defined $*$, $\Subcl(\Sig)$, as a complete bi-Heyting algebra, was already equipped with a canonical intuitionistic negation operation and a canonical paraconsistent negation operation. However, these negations cannot be included in the isomorphism between $E$ and $L$, which motivates the study of $*$ as an independent negation operation on $\Subcl(\Sig)$.\footnote{Note that $\Subcl(\Sig)$, equipped with $*$, is still distributive since we assume the same lattice operations as before.} It is important to note that the Heyting implication operation $\Rightarrow$ will not interact in any nice way with $*$. In order to establish a connection between TQT and QST, we will need to define implication operations on $\Subcl(\Sig)$ that are compatible with $*$. Before doing this, it is useful to prove a couple of lemmas.

From now on, we will stop underlining presheaves to improve readability.
\renewcommand{\Sig}{\Sigma}
\begin{lemma} 
$\eps(S^{*})$ = $\eps(S)^{\bot}$, for any $S \in \Subcl(\Sig)$.
\end{lemma}

\begin{proof} 
$\eps(S^{*}) = \eps(\de(\eps(S)^{\bot})) = \eps(S)^{\bot}$.
\end{proof}

\begin{lemma}
$\de(a)^{*} = \de(a^{\bot})$, for any $a \in \PH$.
\end{lemma}

\begin{proof} 
$\de(a)^{*} = \de(\eps(\de(a))^{\bot}) = \de(a^{\bot})$.
\end{proof}

\begin{lemma}
$\de(a)^{**} = \de(a)$, for any $a \in \PH$, \ie the image of $\de$ is `$*$-regular'.
\end{lemma}

\begin{proof} 
$\de(a)^{**} = \de(a^{\bot})^{*} = \de(a^{\bot \bot}) = \de(a)$.
\end{proof}

\begin{lemma}
$\de(\eps(S))$ is the smallest member of the equivalence class $[S]$ of $S$ under the $\eps$ equivalence relation, and $S^{**}=\de(\eps(S))$ for all $S$.
\end{lemma}

\begin{proof}
Let $T \in [S]$, then $\eps(T) = \eps(S)$. So $\de(\eps(S)) = \de(\eps(T)) \leq T$. Since $T$ was arbitrary, this proves the first part of the lemma. Moreover, we have
\begin{align*}
			S^{**}=\de(\eps(S^*)^{\bot})=\de(\eps(S)^{\bot\bot})=\de(\eps(S)).
\end{align*}
\end{proof}

We are now ready to translate the three quantum material implication operations 
into the distributive setting. We say that an implication operation $\Rightarrow$ on $\Subcl(\Sig)$ `mirrors' a corresponding operation $\rightarrow$ on $L$ if and only if for any $S, T \in \Subcl(\Sig)$ and for any $a, b \in L$, $\de(a \rightarrow b) = \de(a) \Rightarrow \de(b)$ and $\eps(S \Rightarrow T) = \eps(S) \rightarrow \eps(T)$.

\begin{theorem}\label{th:mirror-S} The operation $\Rightarrow_{{\rm S}}$ on $\Subcl(\Sig)$ defined by \[S \Rightarrow_{{\rm S}} T := S^{*} \join (S^{*} \join T^{*})^{*}\] mirrors the Sasaki conditional $\rightarrow_{{\rm S}}$ on $\PH$. 
\end{theorem}

\begin{proof} The assertion follows from the calculations shown below.
\begin{align*}
\de(a \rightarrow_{{\rm S}} b) 
&= \de(a^{\bot} \join (a \meet b)) = \de(a^{\bot}) \join \de(a \meet b) = \de(a)^{*} \join \de(a \meet b)\\
&= \de(a)^{*} \join \de(a \meet b)^{**} = \de(a)^{*} \join \de((a \meet b)^{\bot})^{*}\\
&= \de(a)^{*} \join \de(a^{\bot} \join b^{\bot})^{*} = \de(a)^{*} \join (\de(a^{\bot}) \join \de (b^{\bot}))^{*}\\
&= \de(a)^{*} \join (\de(a)^{*} \join \de (b)^{*})^{*}\\
& = \de(a) \Rightarrow_{{\rm S}} \de(b).
\end{align*}
\begin{align*}
\eps(S \Rightarrow_{{\rm S}} T) 
&= \eps(S^{*} \join (S^{*} \join T^{*})^{*}) = \eps((S \meet (S^{*} \join T^{*}))^{*})\\
&= \eps(S \meet (S^{*} \join T^{*}))^{\bot} = (\eps(S) \meet \eps(S^{*} \join T^{*}))^{\bot} \\
&=\eps(S)^{\bot} \join \eps(S^{*} \join T^{*})^{\bot} = \eps(S)^{\bot} \join \eps((S \meet T)^{*})^{\bot}\\
&= \eps(S)^{\bot} \join \eps(S \meet T)^{\bot \bot} = \eps(S)^{\bot} \join \eps(S \meet T)\\
&= \eps(S) \rightarrow_{{\rm S}} \eps(T). 
\end{align*}
\end{proof}

\begin{theorem}\label{th:mirror-C} The operation $\Rightarrow_{{\rm C}}$ on $\Subcl(\Sig)$ defined by \[S \Rightarrow_{{\rm C}} T := T^{*} \Rightarrow_{{\rm S}} S^{*}\] mirrors the contrapositive Sasaki conditional $\rightarrow_{{\rm C}}$ on $\PH$. 
\end{theorem}

\begin{proof} 
The assertion follows from the calculations shown below.
\begin{align*}
\de(a \rightarrow_{{\rm C}} b) 
&= \de(b^{\bot} \rightarrow_{{\rm S}} a^{\bot}) = \de(b^{\bot}) \Rightarrow_{{\rm S}} \de(a^{\bot}) = \de(b)^{*} \Rightarrow_{{\rm S}} \de(a)^{*}\\ 
&= \de(a) \Rightarrow_{{\rm C}} \de(b).
\end{align*}
\begin{align*}
\eps(S \Rightarrow_{{\rm C}} T) 
&= \eps(T^{*} \Rightarrow_{{\rm S}} S^{*}) = \eps(T^{*}) \rightarrow_{{\rm S}} \eps(S^{*}) = \eps(T)^{\bot} \rightarrow_{{\rm S}} \eps(S)^{\bot}\\
&= \eps(S) \rightarrow_{{\rm C}} \eps(T).
\end{align*}
\end{proof}

\begin{theorem}\label{th:mirror-R} The operation $\Rightarrow_{{\rm R}}$ on $\Subcl(\Sig)$ defined by \[S \Rightarrow_{{\rm R}} T := ((S \Rightarrow_{{\rm S}} T) \meet (S \Rightarrow_{{\rm C}} T))^{**}\] mirrors the relevance conditional $\rightarrow_{{\rm R}}$ on $\PH$. 
\end{theorem}

\begin{proof}
The assertion follows from the calculations shown below.
\begin{align*}
 \de(a \rightarrow_{{\rm R}} b) 
&= \de((a \rightarrow_{{\rm S}} b) \meet (a \rightarrow_{{\rm C}} b)) = \de((a \rightarrow_{{\rm S}} b) \meet (a \rightarrow_{{\rm C}} b))^{**} \\
& = \de((a \rightarrow_{{\rm S}} b)^{\bot} \join (a \rightarrow_{{\rm C}} b)^{\bot})^{*}= (\de((a \rightarrow_{{\rm S}} b)^{\bot}) \join \de((a \rightarrow_{{\rm C}} b)^{\bot}))^{*} \\
&= (\de(a \rightarrow_{{\rm S}} b)^{*} \join \de(a \rightarrow_{{\rm C}} b)^{*})^{*}= (\de(a \rightarrow_{{\rm S}} b) \meet \de(a \rightarrow_{{\rm C}} b))^{**}\\
& = ((\de(a) \Rightarrow_{{\rm S}} \de(b)) \meet (\de(a) \Rightarrow_{{\rm C}} \de(b)))^{**}\\ 
& = \de(a) \Rightarrow_{{\rm R}} \de(b).
\end{align*}
\begin{align*}
\eps(S \Rightarrow_{{\rm R}} T)
&= \eps(((S \Rightarrow_{{\rm S}} T) \meet (S \Rightarrow_{{\rm C}} T))^{**}) = \eps((S \Rightarrow_{{\rm S}} T) \meet (S \Rightarrow_{{\rm C}} T))^{\bot \bot}\\
& = \eps((S \Rightarrow_{{\rm S}} T) \meet (S \Rightarrow_{{\rm C}} T)) = \eps(S \Rightarrow_{{\rm S}} T) \meet \eps(S \Rightarrow_{{\rm C}} T) \\
&= (\eps(S) \rightarrow_{{\rm S}} \eps(T)) \meet (\eps(S) \rightarrow_{{\rm C}} \eps(T)) \\
&= \eps(S) \rightarrow_{{\rm R}} \eps(T).
\end{align*}
\end{proof}

{For $j=\rS,\rC,\rR$ we define the operation $\IFF_j$ on $\Sub$ by  \[S\IFFJ T:= (S\THENJ T)\And(T\THENJ S)\] for any $S,T\in\Sub$.}

We have successfully translated the three quantum material implication operations 
into corresponding operations on the distributive lattice of clopen subobjects of the spectral presheaf equipped with the paraconsistent negation $*$. Of course, the induced implications on $\Subcl(\Sig)$ will have different logical properties to their orthomodular counterparts. For example, they will generally violate modus ponens. This is not surprising given that modus ponens is a problematic inference rule for paraconsistent logics in general (for instance, it is known that modus ponens fails in Priest's \cite{Pri79} paraconsistent `logic of paradox'). At any rate, we can show that all three implication operations share some basic structural properties.

\begin{proposition} \label{Prop_SCRProperties}
For $j = {\rm S}, {\rm C}, {\rm R}$, the operation $\Rightarrow_{j}$ 
satisfies the following properties:
\benum
\item $\top \Rightarrow_{j} S = S^{**}$,
\item $\bot \Rightarrow_{j} S = \top$,
\item $S \Rightarrow_{j} \top = \top$,
\item $S \Rightarrow_{j} \bot = S^*$,
\item $S \Rightarrow_{j} T = \top$ ~if and only if~ $S^{**} \leq T^{**}$,
{\item $S \IFF_{j} T = \top$ ~if and only if~ $S^{**}= T^{**}$.}
\eenum
\end{proposition}

\begin{proof}
 (i)--(iv) follow from the following calculations.
\begin{align*}
\top \Rightarrow_{j} S &=\de(\eps(\top) \rightarrow_{j} \eps(S))=\de(\top\rightarrow_{j} \eps(S))
=\de(\eps(S))=S^{**},\\
\bot \Rightarrow_{j} S&=\de(\eps(\bot) \rightarrow_{j} \eps(S))=\de(\bot \rightarrow_{j}  \eps(S))=\de(\top)=\top,\\
S \Rightarrow_{j} \top&=\de(\eps(S) \rightarrow_{j} \eps(\top))=\de(\eps(S) \rightarrow_{j} \top)=\de(\eps(S)^{\bot} \vee \top)=
\de(\top)=\top,\\
S \Rightarrow_{j} \bot&=\de(\eps(S) \rightarrow_{j} \eps(\bot))=\de(\eps(S) \rightarrow_{j} \bot)=
\de(\eps(S)^{\bot} \vee \bot)=\de(\eps(S)^{\bot})=S^{*}.
\end{align*}

To prove (v), suppose  $S \Rightarrow_{j} T=\top$.  Then  $\eps(S) \rightarrow_{j} \eps(T)=\eps(S \Rightarrow_{j} T)=\eps(\top)=\top$,
so $\eps(S)\le \eps(T)$, and hence $S^{**}=\de(\eps(S)) \le \de(\eps(T)) =T^{**}$.  Conversely, suppose $S^{**}\le T^{**}$.
Then,  $\eps(S)=\eps(S^{**})\le\eps(T^{**})=\eps(T)$, so that $\eps(S \Rightarrow_{j} T)=\eps(S) \rightarrow_{j} \eps(T)=\top$, and hence $S \Rightarrow_{j} T\ge(S \Rightarrow_{j} T)^{**}= \de(\eps(S \Rightarrow_{j} T))=\top$.  Therefore, we have $S \Rightarrow_{j} T=\top$.
{Assertion (vi) follows easily from assertion (v).}
\end{proof}

\subsection{Commutativity and Paraconsistency}

Since if an orthomodular lattice is distributive, it is a Boolean algebra and it represents a classical
logic, it is often considered that a distributive lattice represents a sublogic of the classical logic
or a logic of simultaneously determinate (or commuting) propositions.
However, this idea may conflicts with our embedding of an orthomodular lattice $\PH$,
a typical logic for simultaneously indeterminate (or noncommuting) propositions, 
into a distributive lattice $\Sub$.  In order to resolve this conflict, we introduce the notion of 
commutativity using the paraconsistent negation, and show that noncommutativity 
can be be expressed in a distributive logic with paraconsistent negation.
Thus,  without specifying what negation is considered with it,
it may not be answered whether a distributive lattice represents a classical logic.

\begin{definition} Given $S,T \in \Subcl(\Sig)$, we say that $S$ commutes with $T$ 
(in symbols $S\commutes  T$) if 
$\ep(S)$ commutes with $\ep(T)$ in $\PH$, \ie $\ep(S)=(\ep(S)\And\ep(T))\Or(\ep(S)\And\ep(T)\p)$.
\end{definition}

\bProposition\label{th:com-para}
For any $a, b\in\PH$, we have $a\commutes  b$ if and only if $\de(a)\commutes \de(b)$.
\eProposition
\bProof
Since $\de(a),\de(b)\in\Sub$, we have  $\de(a)\commutes \de(b)$ if and only if
$\ep(\de(a))\commutes \ep(\de(b))$ if and only if $a\commutes  b$.
\eProof
The following theorem shows that the commutativity can be expressed only by lattice
operations and the paraconsistent negation $*$.
 
\bTheorem
For any $S,T \in \Subcl(\Sig)$, we have $S\commutes  T$ if and only if
\begin{align}\label{eq:C-P}
S^{**}=(S^{*}\Or(T^{*}\And T^{**}))^{*}.
\end{align}
\eTheorem
\bProof
Let $S,T \in \Subcl(\Sig)$ and let $f(S,T)=(\ep(S)\And\ep(T))\Or(\ep(S)\And\ep(T)\p)$.
We have 
\begin{align*}
\de\circ f(S,T)
&=\de(\ep(S)\And\ep(T))\Or\de(\ep(S)\And\ep(T)\p)\\
&=\de\circ\ep(S\And T)\Or\de\circ\ep(S\And T^{*})\\
&=(S\And T)^{**}\Or(S\And T^{*})^{**}\\
&=(S^{*}\Or T^{*})^{*}\Or(S^{*}\Or T^{**})^{*}\\
&=((S^{*}\Or T^{*})\And(S^{*}\Or T^{**}))^{*}\\
&=(S^{*}\Or (T^{*}\And T^{**}))^{*}.
\end{align*}
 Suppose $S\commutes  T$.
 Then, we have $\de\circ f(S,T)=\de\circ\ep(S)=S^{**}$, and we obtain relation (\ref{eq:C-P}).
 Conversely, suppose that relation (\ref{eq:C-P}) holds.
 Then, we have $\de\circ f(S,T)=S^{**}=\de\circ\ep(S)$ and hence we have
 $f(S,T)=\ep\circ\de\circ f(S,T)=\ep\circ\de\circ\ep(S)=\ep(S)$, so that $S\commutes  T$ holds.
\eProof

The following theorems show that {the} 
distributive logic $\Sub$ with the negation $*$ 
is properly paraconsistent in the sense that $S\And S^*=\bot$  only if $S=\top$ or $S=\bot$
for all $S\in\Sub$.

\bTheorem	\label{thm:de(a)Andde(a)Star}
For any $a\in \PH$, we have $\de(a)\And\de(a)^*=\bot$  if and only if $a=\top$ or $a=\bot$. 
\eTheorem
\bProof
It is obvious that $\de(a)\And\de(a)^*=\bot$ if  $a=\top$ or $a=\bot$.
Suppose $\de(a)\And\de(a)^*=\bot$.
Let $x\in\PH$.  Then we have
\begin{align*}
(\de(x)^{*}\Or(\de(a)^{*}\And \de(a)^{**}))^{*}
=
(\de(x)^{*}\Or(\de(a)^{*}\And \de(a)))^{*}
=
\de(x)^{**}.
\end{align*}
Thus, $\de(a)\commutes \de(x)$ so that  $a\commutes  x$ by Proposition \ref{th:com-para}.
It follows that $a\commutes  x$ for all $x\in\PH$,
and hence $a=\top$ or $a=\bot$.
\eProof

\bTheorem
For any $S\in \Sub$, we have $S\And S^*=\bot$  if and only if $S=\top$ or $S=\bot$. 
\eTheorem
\bProof
It is obvious that $S\And S^*=\bot$ if  $S=\top$ or $S=\bot$.
Suppose that $S\And S^*=\bot$. Since $S^{**}\le S$, we have $S^{**}\And S^{*}=\bot$. Since $S^{**}=\de(\eps(S))=\de(a)$, where $a=\eps(S)$, we have $S^{*}=S^{***}=\de(a)^{*}$, so $S^{**}\And S^{*}=\de(a)\And\de(a)^{*}=\bot$. By Theorem \ref{thm:de(a)Andde(a)Star}, $a=\bot$ or $a=\top$, so $S^{**}=\de(a)=\bot$ or $S^{**}=\top$. If $S^{**}=\top$, then $S=\top$ since $S^{**}\le S$.  If $S^{**}=\bot$, then $S^{*}=S^{***}
=\top$, so that $S=S\And S^{*}=\bot$.  Thus, if   $S\And S^*=\bot$ then  $S=\top$ or $S=\bot$.
\eProof

The idea now is to study the logical structure of $\Subcl(\Sig)$ equipped with the paraconsistent negation $*$ and the translated orthomodular implications $\Rightarrow_{{\rm S}}$, $\Rightarrow_{{\rm C}}$, $\Rightarrow_{{\rm R}}$ (it turns out that $\Rightarrow_{{\rm S}}$ has a special role here). As we will see in the next section, this new logical structure allows us to develop rich new connections between TQT and QST.

\section{Bridging the Gap}    \label{Sec_Bridge}
\subsection{$V^{(\Subcl(\Sigma))}$}	\label{Sec_VSubclSig}

Eva \cite{Eva15} suggests the possibility of connecting TQT and QST via the set-theoretic structure $V^{(\Subcl(\Sigma))}$, where $\Subcl(\Sigma)$ is equipped with the negation $*$ rather than the Heyting negation. However, he stops short of translating the orthomodular implication connectives into the distributive setting and does not provide any characterisation of the extent to which $V^{(\Subcl(\Sigma))}$ is able to model any interesting mathematics.

\bDefinition
For $j={\rm S},{\rm C},{\rm R}$, we define the $\Subcl(\Sigma)$-valued truth value $\valj{\ph}$ of any formula
$\ph$ in the language of set theory augmented by the names of elements of $\VS$
as the $(\Subcl(\Sigma),\THENJ,\IFF_j)$-interpretation of $\phi$
introduced in definition \ref{intrepretation}.
\eDefinition

\bDefinition
The universe $V$ of the ZFC set theory is embedded into $\VS$ by 
\begin{align*}
\hat{}: V &\rightarrow \VS,\\
x &\mapsto \hat{x}, 
\end{align*}
where $\hat{x} = \{\langle \hat{y}, \top \rangle \: | \: y \in x \} $, \ie 
$\dom({\hat{x}}) = \{ \hat{y} \: | \: y \in x \}$
 and $\hat{x}$ assigns the value $\top$ to every element of its domain. 
\eDefinition

{Note that the relation
 \begin{equation*}
\val{x\in \ck{u}}_j
=\Sup_{u'\in u}\val{\ck{u}'=x}_j
\end{equation*}
follows from the above definition immediately.}
This embedding satisfies the following properties.
\bProposition
The following relations hold for any $u,v\in V$.
\benum
\item $\valj{\ck{u}\in\ck{v}}=\top$  if $u\in v$.
\item $\valj{\ck{u}\in\ck{v}}=\bot$  if $u\not\in v$.
\item $\valj{\ck{u}=\ck{v}}=\top$  if $u=v$.
\item $\valj{\ck{u}=\ck{v}}=\bot$  if $u\not=v$.
\eenum
\eProposition
\bProof
Let $u,v\in{V}$.  Suppose that relations (iii) and (iv) hold for any
$u'\in{u}$ and $v'\in{v}$.  
\begin{align*}
\valj{\ck{u'}\in\ck{v}}&=\Sup_{{v''}\in\dom({\ck{v}})}{\ck{v}}({v''})\And \valj{\ck{u'}={v''}}
=\Sup_{v'\in v}
{\ck{v}}(\ck{v'})\And\valj{\ck{u'}=\ck{v'}}
=\Sup_{v'\in v}\valj{\ck{u'}=\ck{v'}},\\
\valj{\ck{u}=\ck{v}}&=
\Inf_{{u''}\in\dom({\ck{u}})}{\ck{u}}({u''})\THENJ \valj{{u''}\in {\ck{v}}}\And
\Inf_{{v''}\in\dom({\ck{v}})}{\ck{v}}({v''})\THENJ \valj{{v''}\in {\ck{u}}}\\
&=\Inf_{u'\in u}{\ck{u}}(\ck{u'})\THENJ \valj{\ck{u'}\in {\ck{v}}}\And
\Inf_{v'\in v}{\ck{v}}(\ck{v'})\THENJ \valj{\ck{v'}\in {\ck{u}}}\\
&=
{
\Inf_{u'\in u}\top\THENJ \valj{\ck{u'}\in {\ck{v}}}\And
\Inf_{v'\in v}\top\THENJ \valj{\ck{v'}\in {\ck{u}}}}\\
&=\Inf_{u'\in u}\valj{\ck{u'}\in \ck{v}}^{**}\And
\Inf_{v'\in v}\valj{\ck{v'}\in \ck{u}}^{**}\\
&=\Inf_{u'\in u}(\Sup_{v'\in v}\valj{\ck{u'}=\ck{v'}})^{**}
\And
\Inf_{v'\in v}(\Sup_{u'\in u}\valj{\ck{v'}=\ck{u'}})^{**},
\end{align*}
{where Proposition \ref{Prop_SCRProperties} (i) was used in the penultimate equality. }
Thus, $\valj{\ck{u}=\ck{v}}=\top$ if $u=v$, and 
$\valj{\ck{u}=\ck{v}}=\bot$ if $u\not=v$ for all $u,v\in{V}$ 
{from the relations $\top^{**}=\top$ and $\bot^{**}=\bot$.}
Consequently, relations (iii) and (iv) have been proved by induction.
Then, relations (i) and (ii) follow straightforwardly.
\eProof

\bCorollary
The following statements hold.
\begin{enumerate}[{\rm (1)}]
\item
Given a $\De_{0}$-formula $\phi$ with $n$ free variables, and $x_{1},...,x_{n} \in V$, $\phi (x_{1}, ... , x_{n}) \leftrightarrow 
\VS \models \phi (\hat{x}_{1}, ... , \hat{x}_{n})$. 
\item
Given a $\Sigma_{1}$-formula $\phi$ with $n$ free variables, and $x_{1},...,x_{n} \in V$, $\phi (x_{1}, ... , x_{n}) \rightarrow 
\VS \models \phi (\hat{x}_{1}, ... , \hat{x}_{n})$. 
\end{enumerate}
\eCorollary

The first thing to note is that, assuming $*$, $V^{(\Subcl(\Sigma))}$ is a set-theoretic structure with a \emph{paraconsistent} internal logic. In recent years, Weber \cite{Web10,Web12}, Brady \cite{Bra89}, L{\"o}we and Tarafder \cite{LowTar15} and others have done exciting work in exploring the possibility of developing a non-trivial set theory built over a paraconsistent logic. However, there is still no well established model theory for paraconsistent set theory, despite some promising recent developments (see \eg Libert \cite{Lib05}, L{\"o}we and Tarafder \cite{LowTar15}). So it is not currently possible to characterise $V^{(\Subcl(\Sigma))}$ as a full model of any particular set theory. However, we will now show that it is possible to translate Ozawa's $\De_{0}$ transfer principle for orthomodular valued models to the paraconsistent structure $V^{(\Subcl(\Sigma))}$. This guarantees that $V^{(\Subcl(\Sigma))}$ is able to model significant fragments of classical mathematics. The following definition will play an important role. 

\begin{definition} The maps $\alpha: V^{(\PH)} \rightarrow V^{(\Subcl(\Sig))}$ and 
$\omega: V^{(\Subcl(\Sig))} \rightarrow V^{(\PH)}$ are given by the following recursive definitions:
\begin{enumerate}[{\rm (1)}]
\item Given $u \in V^{(\PH)}$, $\alpha(u) = \{\langle \alpha(x), \de(u(x))\rangle| x \in \dom(u)\}$.
\item Given $u \in V^{(\Subcl(\Sig))}, \omega(u) = \{\langle \omega(x), \eps(u(x))\rangle| x \in \dom(u)\}$.
\end{enumerate}
\end{definition}

Intuitively, $\alpha$ is an embedding of the orthomodular valued structure $V^{(\PH)}$ into the paraconsistent and distributive structure $V^{(\Subcl(\Sig))}$. It allows us to translate the constructions of QST into a new setting whose logic is closely connected to TQT. $\alpha$ and $\omega$ can be seen as `higher level' versions of the morphisms $\de$, $\eps$ that map between whole Q-worlds rather than simple lattices.  

Before proving a fundamental theorem concerning the implications of $\alpha$, recall the induction principle for algebraic valued models of set theory, which says that for any algebra $A$,
\[
{\forall u \in V^{(A)}(\forall u' \in \dom(u)\phi(u') \rightarrow \phi(u)) \rightarrow \forall u \in V^{(A)}\phi(u),}
\]
\ie if we want to prove that $\phi$ holds for every element of the structure $V^{(A)}$, it is sufficient to show that for arbitrary $u \in V^{(A)}$, $\phi$ holding for everything in $u$'s domain implies that $\phi$ holds for $u$. We use this inductive principle in the proof of the following key result. 
In the following, a `negation-free $\De_{0}$-formula' is any formula constructed from atomic formulae of the form $x=y$ or $x\in y$ by adding conjunction $\wedge$, disjunction $\vee$, and bounded quantifiers $(\forall x\in y)$ and $(\exists x\in y)$, where $x$ and $y$ denote arbitrary variables.

\begin{theorem}	\label{Thm_deltaalpha}
For any negation-free $\De_{0}$-formula $\phi(x_{1}, ..., x_{n})$ and any $u_{1}, ..., u_{n} \in V^{(\PH)}$, 
\[
\de(\|\phi(u_{1}, ..., u_{n})\|_{{\rm S}}) \leq \|\phi(\alpha(u_{1}), ..., \alpha(u_{n}))\|_{{\rm S}}.
\]
\end{theorem}

\begin{proof} In the proof, we omit the symbol $\rS$ and simply assume that $\Rightarrow$ always denotes $\Rightarrow_{{\rm S}}$. 
Argue by induction. Let $u \in V^{(\PH)}$.
{
We begin with proving the following relations for atomic formulas.
\bigskip

{\centering
\begin{varwidth}{\textwidth}
\begin{enumerate}\item[(i)] $\|\alpha(u) = \alpha(v)\| \geq \de(\|u = v\|)$,
\item[(ii)] $\|\alpha(u) \in \alpha(v)\| \geq \de(\|u \in v\|)$,
\item[(iii)] $\|\alpha(v) \in \alpha(u)\| \geq \de(\|v \in u\|)$ for any $v \in V^{(\PH)}$.
\end{enumerate}
\end{varwidth}
\par}
\bigskip

\noindent
Suppose $u' \in \dom(u)$.  By induction hypothesis we have

\bigskip
{\centering
\begin{varwidth}{\textwidth}
\begin{enumerate}\item[(i')] $\|\alpha(u') = \alpha(v)\| \geq \de(\|u' = v\|)$,
\item[(ii')] $\|\alpha(u') \in \alpha(v)\| \geq \de(\|u' \in v\|)$ for any $v \in V^{(\PH)}$.
\end{enumerate}
\end{varwidth}
\par}

\bigskip

\noindent
Let  $v \in V^{(\PH)}$.  It follows from (i') that 
\deqs{
\|\alpha(v) \in \alpha(u)\|
&= \bjoin_{u' \in \dom(u)}( \alpha(u)(\alpha(u')) \meet \|\alpha(u') = \alpha(v)\|)\\
&= \bjoin_{u' \in \dom(u)} (\de(u(u')) \meet \|\alpha(u') = \alpha(v)\|)\\
&\geq \bjoin_{u' \in \dom(u)} [\de(u(u')) \meet \de(\|u' = v\|)]\\
&\geq \bjoin_{u' \in \dom(u)} \de(u(u') \meet \|u' = v\|)\\
&= \de\left(\bjoin_{u' \in \dom(u)} (u(u') \meet \|u' = v\|)\right)\\
&= \de(\|v \in u\|).
}
So we have shown (iii) for any $v\in \VPH$.

Using the easily observed fact that $b \leq c$ implies $a \Rightarrow_{{\rm S}} b \leq a \Rightarrow_{{\rm S}} c$, as well as  (ii') for an arbitrary $v \in V^{(\PH)}$ and (i) with substituting an arbtrary $v' \in \dom(v)$ for $v$,  we have
\deqs{
\lefteqn{\|\alpha(u) = \alpha(v)\|}\quad\\ 
&= \bmeet_{u' \in \dom(u)} (\alpha(u)(\alpha(u')) \Rightarrow \|\alpha(u') \in \alpha(v)\| )\meet \bmeet_{v' \in \dom(v)}( \alpha(v)(\alpha(v')) \Rightarrow \|\alpha(v') \in \alpha(u)\|)\\
&= \bmeet_{u' \in \dom(u)} (\de(u(u')) \Rightarrow \|\alpha(u') \in \alpha(v)\|) \meet \bmeet_{v' \in \dom(v)} (\de(v(v')) \Rightarrow \|\alpha(v') \in \alpha(u)\|)\\
&\geq \bmeet_{u' \in \dom(u)}[ \de(u(u')) \Rightarrow \de(\|u' \in v\|)] \meet \bmeet_{v' \in \dom(v)}[ \de(v(v')) \Rightarrow \de(\|v' \in u\|)]\\
&= \bmeet_{u' \in \dom(u)} \de(u(u') \rightarrow \|u' \in v\|) \meet \bmeet_{v' \in \dom(v)} \de(v(v') \rightarrow \|v' \in u\|)\\
&\geq \de\left(\bmeet_{u' \in \dom(u)} \left(u(u') \rightarrow \|u' \in v\|\right) \meet \bmeet_{v' \in \dom(v)}\left(v(v') \rightarrow \|v' \in u\|\right)\right)\\
&=\|u=v\|.
}
So we have shown (i) for any $v\in\VPH$.

It remains to show (ii).  Let $v\in\VPH$.
It follows from  (i) with substituting an arbitrary $v' \in \dom(v)$ for $v$ that
\deqs{
\| \alpha(u) \in \alpha(v)\| 
&= \bjoin_{v' \in \dom(v)} \left(\alpha(v)(\alpha(v')) \meet \|\alpha(u) = \alpha(v')\|\right)\\
&\geq \bjoin_{v' \in \dom(v)} [\de(v(v')) \meet \de(\|u = v'\|)]\\
&\geq \bjoin_{v' \in \dom(v)}\de( v(v') \meet \|u = v'\|)\\
&= \de\left(\bjoin_{v' \in \dom(v)} (v(v') \meet \|u = v'\|)\right)\\
&= \de(\|u \in v\|).
}
So we have derived (ii) for any $v\in\VPH$.
}
Thus, the assertion holds for all atomic formulae. 

Suppose that negation-free $\De_{0}$-formulae $\phi_j(\vec{x})$ for $j=1,2$ satisfy 
\[
\de(\|\phi_j(\vec{u})\|) \leq \|\phi_j(\alpha(\vec{u})))\|
\]
for any $u_1,\ldots,u_n\in V^{(\PH)}$, where $\vec{x}=(x_1,\ldots,x_n)$, $\vec{u}=(u_1,\ldots,u_n)$, and $\alpha(\vec{u})=(\alpha(u_1),\ldots,\alpha(u_n))$.  Then we have
\begin{align*}
\|\phi_1(\alpha(\vec{u}))\meet\phi_2(\alpha(\vec{u}))\|
&=\|\phi_1(\alpha(\vec{u}))\|\meet \|\phi_2(\alpha(\vec{u}))\|\\
&\geq \de(\|\phi_1(\vec{u})\|)\meet \de(\|\phi_2(\vec{u})\|)\\
&\geq \de(\|\phi_1(\vec{u})\|\meet\|\phi_2(\vec{u})\|)\\
&=\de(\|\phi_1(\vec{u})\meet\phi_2(\vec{u})\|).
\end{align*}
We also have 
\begin{align*}
\|\phi_1(\alpha(\vec{u}))\join\phi_2(\alpha(\vec{u}))\|
&=\|\phi_1(\alpha(\vec{u}))\|\join \|\phi_2(\alpha(\vec{u}))\|\\
&\geq \de(\|\phi_1(\vec{u})\|)\join \de(\|\phi_2(\vec{u})\|)\\
&=\de(\|\phi_1(\vec{u})\|\join\|\phi_2(\vec{u})\|)\\
&=\de(\|\phi_1(\vec{u})\join\phi_2(\vec{u})\|).
\end{align*}

Suppose that a negation-free $\De_{0}$-formula $\phi(x,\vec{x})$ satisfies
\[
\de(\|\phi(u,\vec{u})\|) \leq \|\phi(\alpha(u),\alpha(\vec{u}))\|
\]
for any $u,u_1,\ldots.u_n \in V^{(\PH)}$. Then, we have
\begin{align*}
\|(\forall x\in \alpha(u))\phi(x,\alpha(\vec{u}))\|
&=\bmeet_{w\in \dom(\alpha(u))}(\alpha(u)(w)\Rightarrow\|\phi(w,\alpha(\vec{u}))\|)\\
&=\bmeet_{u'\in \dom(u)}(\alpha(u)(\alpha(u'))\Rightarrow\|\phi(\alpha(u'),\alpha(\vec{u}))\|)\\
&\geq\bmeet_{u'\in \dom(u)}[\alpha(u)(\alpha(u'))\Rightarrow\de(\|\phi(u',\vec{u})\|)]\\
&=\bmeet_{u'\in \dom(u)}[\de(u(u'))\Rightarrow\de(\|\phi(u',\vec{u})\|)]\\
&{=\bmeet_{u'\in \dom(u)}\de(u(u')\rightarrow\|\phi(u',\vec{u})\|)}\\
&\geq\de\left(\bmeet_{u'\in \dom(u)}(u(u')\rightarrow\|\phi(u',\vec{u})\|)\right)\\
&\geq\de(\|(\forall x\in u)\phi(x,\vec{u})\|),
\end{align*}
and we also have
\begin{align*}
\|(\exists x\in \alpha(u))\phi(x,\alpha(\vec{u})\|
&=\bjoin_{w\in \dom(\alpha(u))}(\alpha(u)(w)\meet\|\phi(w,\alpha(\vec{u})\|)\\
&=\bjoin_{u'\in \dom(u)}(\alpha(u)(\alpha(u'))\meet\|\phi(\alpha(u'),\alpha(\vec{u}))\|)\\
&\geq\bjoin_{u'\in \dom(u)}[\alpha(u)(\alpha(u'))\meet\de(\|\phi(u',\vec{u})\|)]\\
&=\bjoin_{u'\in \dom(u)}[\de(u(u'))\meet\de(\|\phi(u',\vec{u})\|)]\\
&\geq\bjoin_{u'\in \dom(u)}\de(u(u')\meet\|\phi(u',\vec{u})\|)\\
&=\de\left(\bjoin_{u'\in \dom(u)}(u(u')\meet\|\phi(u',\vec{u})\|)\right)\\
&=\de(\|(\exists x\in u)\phi(x,\vec{u})\|).
\end{align*}

This completes the proof by induction on the complexity of negation-free $\De_{0}$-formulae.
\end{proof}

{The above proof uses the monotonicity property of the Sasaki arrow:
$b \leq c$ implies $a \Rightarrow_{{\rm S}} b \leq a \Rightarrow_{{\rm S}} c$. 
In the case where $j={\rm R}$ or ${\rm C}$, the corresponding property does not hold,
so that the proof does not work in those cases.}
 
Combined with theorem \ref{Thm_ZFCTransfer} and the monotonicity of $\de$, theorem \ref{Thm_deltaalpha} immediately leads to the following important corollary,
{for which we recall definition \ref{def:commutator}.}

\begin{theorem}\label{Thm_ZFCdeltaalpha} For any negation-free $\De_{0}$-formula $\phi(x_{1}, ..., x_{n})$ and any $u_{1}, ..., u_{n} \in V^{(\PH)}$, if $\phi(x_{1}, ..., x_{n})$ is provable in ZFC then
\[
\de(\ul{\join} (u_{1}, ..., u_{n})) \leq \|\phi(\alpha(u_{1}), ... \alpha(u_{n}))\|_{{\rm S}}.
\]
\end{theorem}

Equipped with theorem \ref{Thm_ZFCdeltaalpha}, we can guarantee that $V^{(\Subcl(\Sigma))}$ will model significant fragments of classical mathematics. Crucially, we can now translate many of Ozawa's \cite{Oza07} results characterising the behaviour of the real numbers in orthomodular valued models to the paraconsistent structure $V^{(\Subcl(\Sigma))}$.\footnote{It is also worth briefly pointing out that $V^{(\Subcl(\Sigma))}$ is very clearly a non-trivial structure. It is not the case that every $\Subcl(\Sigma)$-sentence is satisfied by $V^{(\Subcl(\Sigma))}$. To give a simple example, let $e \in V^{(\Subcl(\Sigma))}$ be any element such that for every $x \in \dom(e)$, $e(x) = \bot$. Then it is easy to see that for any $u \in V^{(\Subcl(\Sigma))}$, $\|u \in e\| = \bot$.}

\subsection{Real Numbers in $V^{(\Subcl(\Sigma))}$ }
We are now ready to use the paraconsistent structure $V^{(\Subcl(\Sigma))}$ to establish a rich and powerful connection between TQT and QST. First of all, we can define the Dedekind reals in $V^{(\Subcl(\Sigma))}$ in the usual way. Eva \cite{Eva15} shows that, in the case where the original orthomodular lattice is a projection lattice $\PH$, {the set} $\mathbb{R}^{(\Subcl(\Sigma))}$ of all Dedekind reals in $V^{(\Subcl(\Sigma))}$
is closely related to the set $SA(\cH)$ of all self-adjoint operators on $\cH$.
{Here, we render this connection more precisely. }

Recall first that $\mathbb{R}^{(\Subcl(\Sigma))}$ is defined by 

\bc
$\mathbb{R}^{(\Subcl(\Sigma))} = \{u \in V^{(\Subcl(\Sigma))}| \dom(u) = \dom( \hat{\mathbb{Q}}) \meet \valj{\mathbf{R}(u)} = \top\}$,
\ec
\noindent where $\mathbf{R}(u)$ is the formula 

\vspace{-0.4cm}

\begin{align*}
\mathbf{R}(u) := \forall y \in u (y \in \hat{\mathbb{Q}}) \meet &\exists y \in \hat{\mathbb{Q}} (y \in u) \meet 
\exists y \in \hat{\mathbb{Q}} (y \notin u) \meet\\
&\forall y \in \hat{\mathbb{Q}} (y \in u \leftrightarrow 
\exists z \in \hat{\mathbb{Q}} (z < y \meet z \in u)),
\end{align*}
meaning that $u$ is the upper segment of a Dedekind cut of the rational numbers without endpoint.
\begin{proposition}
{Let  $u,v\in\RS$. The following statements hold.}
\begin{enumerate}[{\rm (1)}]
\item
$\valj{\hat{r} \in u} = u(\hat{r})$
{ for all $r \in \mathbb{Q}$.}
{
\item 
$\valj{u=v}=\top$ if and only if 
$u(\ck{r})^{**}=v(\ck{r})^{**}$ for all $r\in\Q$.}
\end{enumerate}
\end{proposition}

\begin{proof}
Statement (1) follows from
\[
\valj{\hat{r} \in u} = \bigvee \limits_{s \in \mathbb{Q}} u(\hat{s}) \wedge \valj{\hat{s} = \hat{r}} = u(\hat{r}).
\]
{
To show statement (2), let $u,v\in\RS$.  We obtain
\deqs{
\valj{u=v}
&=
\Inf_{u'\in\dom(u)}
(u(u')\THENJ\valj{u'\in v})\And
\Inf_{v'\in\dom(v)}(v(v')\THENJ\valj{v'\in u})\\
&=
\Inf_{u'\in\dom(\ck{\mathbb{Q}})}
(u(u')\THENJ\valj{u'\in v})\And
\Inf_{v'\in\dom(\ck{\mathbb{Q}})}(v(v')\THENJ\valj{v'\in u})\\
&=
\Inf_{r\in\mathbb{Q}}
(u(\ck{r})\THENJ\valj{\ck{r}\in v})\And
\Inf_{r\in\mathbb{Q}}(v(\ck{r})\THENJ\valj{\ck{r}\in u})\\
&=
\Inf_{r\in\mathbb{Q}}
[(u(\ck{r})\THENJ\valj{\ck{r}\in v})\And
(v(\ck{r})\THENJ\valj{\ck{r}\in u})]\\
&=
\Inf_{r\in\mathbb{Q}}
[(u(\ck{r})\THENJ v(\ck{r}))\And
(v(\ck{r})\THENJ u(\ck{r}))]\\
&=
\Inf_{r\in\mathbb{Q}}
(u(\ck{r})\IFF_j v(\ck{r})).
}
Thus, $\valj{u=v}=\top$ if and only if 
$u(\ck{r})\IFF_j v(\ck{r})=\top$ for all $r\in\Q$
if and only if $u(\ck{r})^{**}=v(\ck{r})^{**}$ for all $r\in\Q$
by proposition \ref{Prop_SCRProperties} (vi).
}
\end{proof}

\begin{proposition}	\label{Prop_uLikeSpecFam}
For any $u \in V^{(\Subcl(\Sigma))}$ with $dom(u) = \hat{\mathbb{Q}}$, {we have} 
$u \in \mathbb{R}^{(\Subcl(\Sigma))}$ if and only if the following conditions hold:

\bigskip
{\centering
\begin{varwidth}{\textwidth}
\begin{enumerate}[{\rm (i)}]
\item$ \bigvee \limits_{r \in \mathbb{Q}} u(\hat{r}) = \top$.
\item$ \bigvee \limits_{r \in \mathbb{Q}} u(\hat{r})^* = \top$.
\item$ \left(\bigvee \limits_{s \in \mathbb{Q}: s < r} u(\hat{s})\right)^{**} = u(\hat{r})^{**}$
{for all $r\in\Q$.}
\end{enumerate}
\end{varwidth}
\par}
\bigskip
\end{proposition}

\begin{proof}
By definition, $u \in \mathbb{R}^{(\Subcl(\Sigma))}$ if and only if $\valj{\mathbf{R}(u)} = \top$ if and only if the following conditions hold:

\bigskip
{\centering
\begin{varwidth}{\textwidth}
\begin{enumerate}  
\item[(P1)]$ \valj{\exists y \in \hat{\mathbb{Q}} (y \in u)} = \top$.
\item[(P2)]$ \valj{\exists y \in \hat{\mathbb{Q}} (y \notin u)} = \top$.
\item[(P3)]$ \valj{\forall y \in \hat{\mathbb{Q}} [y \in u \leftrightarrow \exists z \in \hat{\mathbb{Q}} (z < y \wedge z \in u)]} = \top$.
\end{enumerate}
\end{varwidth}
\par}
\bigskip
We have that 
\bc
$\valj{\exists y \in \hat{\mathbb{Q}} (y \in u)} = \bigvee \limits_{r \in \mathbb{Q}} \hat{\mathbb{Q}}(\hat{r}) \wedge \valj{\hat{r} \in u} = \bigvee \limits_{r \in \mathbb{Q}}  \valj{\hat{r} \in u} = \bigvee \limits_{r \in \mathbb{Q}}  u(\hat{r})$
\ec
\bc
$\valj{\exists y \in \hat{\mathbb{Q}} (y \notin u)} = \bigvee \limits_{r \in \mathbb{Q}} \hat{\mathbb{Q}}(\hat{r}) \wedge \valj{\hat{r} \notin u} = \bigvee \limits_{r \in \mathbb{Q}}  \valj{\hat{r} \in u}^{*} = \bigvee \limits_{r \in \mathbb{Q}}  u(\hat{r})^{*}$
\ec
Thus, (P1) $\Leftrightarrow$ (i) and (P2) $\Leftrightarrow$ (ii). Furthermore, 
{
\deqs{
\lefteqn{
\valj{\forall y\in \ck{\mathbb{Q}}
[y\in u \Iff
\exists z\in \ck{\mathbb{Q}} (z<y\And z\in u)]}}\qquad\\
&={\Inf_{r\in\mathbb{Q}}\ck{\mathbb{Q}}(\ck{r})\THENJ
\left(\valj{\ck{r}\in u}\IFFJ
\Sup_{s\in\Q}(\ck{\Q}(\ck{s})\And\valj{\ck{s}<\ck{r}}\And\valj{\ck{s}\in u})\right)}\\
&={\Inf_{r\in\mathbb{Q}}\top\THENJ
\left(\valj{\ck{r}\in u}\IFFJ
\Sup_{s\in\Q}(\top\And\valj{\ck{s}<\ck{r}}\And\valj{\ck{s}\in u})\right)}\\
&={\Inf_{r\in\mathbb{Q}}
\left(u(\ck{r})\IFFJ
\Sup_{s\in\Q}(\valj{\ck{s}<\ck{r}}\And u(\ck{s}))\right)^{**}}\\
&={\Inf_{r\in\mathbb{Q}}
\left(u(\ck{r})\IFFJ
\Sup_{s\in\Q:s<r}u(\ck{s})\right)^{**}}.
}
Since $S^{**}=\top \IFFJ S=\top$,  (P3) holds if and only if  $[u(\ck{r})\IFFJ
\Sup_{s\in\Q:s<r}u(\ck{s})]=\top$ holds for all $r\in\Q$.
And since $S \Leftrightarrow T = \top$ if and only if $S^{**}=T^{**}$ by proposition \ref{Prop_SCRProperties} (vi), 
}
it follows that (P3) $\Leftrightarrow$ (iii), as desired. 
\end{proof}

\begin{definition}
Define the map $H: SA(\cH) \rightarrow V^{(\Subcl(\Sigma))}$ by dom$(H(X)) = \hat{\mathbb{Q}}$ and 
\[
H(X)(\hat{r}) = \delta(E^{X}_{r})
\]
for all $r \in \mathbb{Q}$.
\end{definition}

Intuitively, the map $H$ allows us to represent self-adjoint operators within the structure $V^{(\Subcl(\Sigma))}$. Since $\delta$ is injective and the correspondence between left-continuous spectral families and self-adjoint operators is bijective, it follows that $H$ is also injective. 

\begin{proposition}	\label{Prop_XLikeSpecFamPlusRegular}
For any $X \in SA(\cH)$, the following properties hold. 

\bigskip
{\centering
\begin{varwidth}{\textwidth}
\begin{enumerate}[{\rm (i)}]
\item
$ \bigvee \limits_{r \in \mathbb{Q}} H(X)(\hat{r}) = \top$,
\item
$ \bigvee \limits_{r \in \mathbb{Q}} H(X)(\hat{r})^* = \top$,
\item
$ \bigvee \limits_{s \in \mathbb{Q}: s < r} H(X)(\hat{s}) = H(X)(\hat{r})$,
\item
$H(X)(\hat{r}) = H(X)(\hat{r})^{**}$, $\forall r \in \mathbb{Q}$.
\end{enumerate}
\end{varwidth}
\par}
\bigskip
\end{proposition}
\begin{proof}
Properties (i) and (iii) follow immediately from the defining properties of left-continuous spectral families and the fact that $\delta$ preserves joins. Property (ii) follows from join preservation and the fact that $\delta(a^{\bot}) = \delta(a)^{*}$. To prove (iv), note that since the image of $\delta$ is *-regular (Lemma 15), we have
\bc
$H(X)(\hat{r})^{**} = \delta(E^{X}_{r})^{**} = \delta(E^{X}_{r}) = H(X)(\hat{r})$.
\ec
\end{proof}

\begin{definition}
For any $u \in V^{(\Subcl(\Sigma))}$, define $u^*$ by dom$(u^*)$ = dom$(u)$ and $u^*(v) = u(v)^*$, for all $v \in $dom$(u)$. We call u$^*$ the \emph{complement of $u$}, and $u^{**}$ the \emph{regularisation of $u$}. A real number $u \in \mathbb{R}^{(\Subcl(\Sigma))}$ is called \emph{regular} if $u^{**} = u$. Denote by $\RSR$ the set of regular elements of $\mathbb{R}^{(\Subcl(\Sigma))}$.
\end{definition}

\begin{proposition}\label{th:daseinisation-H}
The mapping $H:X \mapsto H(X)$ maps $SA(\cH)$ to $\RSR$ , \ie $H[SA(\cH)] \subseteq \RSR$. 
\end{proposition}
\begin{proof}
From proposition \ref{Prop_uLikeSpecFam} and proposition \ref{Prop_XLikeSpecFamPlusRegular} (i)-(iii), we have that $H[SA(\cH)] \subseteq \mathbb{R}^{(\Subcl(\Sigma))}$. Proposition \ref{Prop_XLikeSpecFamPlusRegular} (iv) shows that $H[SA(\cH)] \subseteq \RSR$.
\end{proof}

\begin{proposition}	\label{Prop_epsPlaysNicelyWithRegularJoins}
For any family $\{S_{j}\} \subseteq \Subcl(\Sigma)$ satisfying $S_{j}^{**} = S_{j}$ for all $j$, \
\bc
$\varepsilon(\bigvee_{j} S_{j}) = \bigvee_{j} \varepsilon(S_{j})$.
\ec
\end{proposition}
\begin{proof}
Let $\{S_{j}\} \subseteq \Subcl(\Sigma)$ satisfy the condition. Then 

\begin{align*}
&\varepsilon(\bigvee_{j} S_{j}) =\varepsilon(\bigvee_{j} S_{j}^{**}) = \varepsilon((\bigwedge_{j} S_{j}^{*})^{*}) = \varepsilon(\bigwedge_{j} S_{j}^{*})^{\bot}\\
=\; &(\bigwedge_{j} \varepsilon(S_{j}^{*}))^{\bot} = (\bigwedge_{j} \varepsilon(S_{j})^{\bot})^{\bot} = (\bigvee_{j} \varepsilon(S_{j}))^{\bot \bot} = \bigvee_{j} \varepsilon(S_{j}).
\end{align*}

\end{proof}

\begin{definition}
Given $u \in \RSR$, 
define the $P(\cH)$-valued function $F^{u}$ on {$\mathbb{Q}$} by
\[
F^{u}(r) = \varepsilon(u(\hat{r})).
\]
{for all $r\in\Q$.}
\end{definition}

\begin{proposition}	\label{Prop_PropertiesOfF}
For any $u \in \RSR$
the following properties hold: 

\bigskip
{\centering
\begin{varwidth}{\textwidth}
\begin{enumerate}[{\rm (i)}]
\item
$\bigvee \limits_{r \in \mathbb{Q}} F^{u}(r) = \top.$
\item
$ \bigwedge \limits_{r \in \mathbb{Q}} F^{u}(r) = \bot.$
\item
$ \bigvee \limits_{s \in \mathbb{Q}: s < r} F^{u}(s) = F^{u}(r)$
\quad for all $r \in \mathbb{Q}$.
\end{enumerate}
\end{varwidth}
\par}
\end{proposition}
\begin{proof}
Let $u \in \RSR$ and $r \in \mathbb{Q}$. Since $u(\hat{r})$ is regular, it follows (from proposition \ref{Prop_epsPlaysNicelyWithRegularJoins}) that
\begin{align*}
\text{(i) } &\bigvee \limits_{r \in \mathbb{Q}} F^{u}(r) = \bigvee \limits_{r \in \mathbb{Q}} \varepsilon(u(\hat{r})) = \varepsilon(\bigvee \limits_{r \in \mathbb{Q}} u(\hat{r})) = \varepsilon(\top) = \top,\\
\text{(ii) } &\bigwedge \limits_{r \in \mathbb{Q}} F^{u}(r) = \bigwedge \limits_{r \in \mathbb{Q}} \varepsilon(u(\hat{r})) = (\bigvee \limits_{r \in \mathbb{Q}} \varepsilon(u(\hat{r}))^{\bot})^{\bot} =(\bigvee \limits_{r \in \mathbb{Q}} \varepsilon(u(\hat{r})^{*}))^{\bot}\\
=\; &\varepsilon(\bigvee \limits_{r \in \mathbb{Q}} u(\hat{r})^{*})^{\bot} = \eps(\top)^{\bot} = \bot,\\
\text{(iii) } &\bigvee \limits_{s \in \mathbb{Q}: s < r} F^{u}(s) = \bigvee \limits_{s \in \mathbb{Q}: s < r} \eps(u(\hat{s})) =  \eps(\bigvee \limits_{s \in \mathbb{Q}: s < r} u(\hat{s})) = \eps(u(\hat{r})) = F^{u}(r).
\end{align*}
\end{proof}

\begin{proposition}	\label{Thm_uGivesSpecFam}
Each $u \in \RSR$ uniquely determines a corresponding self-adjoint operator $G(u) \in SA(\cH)$ defined by the left-continuous spectral family
{
\[E^{G(u)}_{\la} = \Sup_{r\in\Q:r<\la}F^{u}(r)\] 
for each $\la \in {\mathbb{R}}$.
}
\end{proposition}
\begin{proof}
That $\{E^{G(u)}_{\la}|\la \in \mathbb{{R}}\}$ is a left-continuous spectral family follows from proposition \ref{Prop_PropertiesOfF}, and the spectral theorem entails that this family uniquely defines a corresponding self-adjoint operator. 
\end{proof}

{
\bProposition\label{th:P25-6}
The following relations hold.
\benum
\item $G\circ H(A)=A$\quad for all $A\in SA(\cH)$.
\item $G[\RSR]=SA(\cH)$.
\item $H\circ  G(u)=u$\quad for all $u\in\RSR$.
\item $H[SA(\cH)]=\RSR$.
\eenum
\eProposition
\bProof
Let $A\in SA(\cH)$ and $r\in\Q$.  Then we have
\[
E^{G\circ H(A)}_r=F^{H(A)}(r)=\ep(H(A)(\hat{r}))=
\ep\circ \de(E^{A}_r)=E^{A}_r,
\]
and hence assertion (ii) holds. From proposition \ref{Thm_uGivesSpecFam}
 we have $G[\RSR]\subseteq SA(\cH)$
and (i) concludes assertion (ii).
Let $u\in\RSR$ and $r\in\Q$.  Then we have
\[
H\circ  G(u)(\ck{r})=\de(E^{G(u)}_r)=\de\circ \ep(u(\ck{r}))=u(\ck{r})^{**}=u(\ck{r}),
\]
and hence assertion (iii) holds. 
From proposition \ref{th:daseinisation-H},
in order to show (iv)  it suffices to show the relation 
$\RSR\subseteq H[SA(\cH)]$.
Let $u\in\RSR$ and $A=G(u)\in SA(\cH)$. Then we have
\[
H(A)(\ck{r})=\de(E^{G(u)}_r)=\de\circ \ep(u(\ck{r}))=u(\ck{r})^{**}=u(\ck{r})
\]
for all $r\in\Q$.  Therefore, the relation $H[SA(\cH)]=\RSR$ holds.
\eProof
}

Propositon \ref{Thm_uGivesSpecFam} shows {that mutually inverse mappings
$G$ and $H$ establish a one-to-one correspondence between $SA(\cH)$ and $\RSR$,
so} that we can faithfully represent all regular reals in $V^{(\Subcl(\Sigma))}$ as self-adjoint operators. 
{The following theorem summarizes the above results.}

\begin{theorem}	\label{Thm_SAOPsBijectiveRegReals}
$SA(\cH)$ and $\RSR$ are in bijective correspondence
{under mutually inverse bijections $G$ and $H$.} 
\end{theorem}

Theorem \ref{Thm_SAOPsBijectiveRegReals} provides a precise clarification of the relationship between real numbers in $V^{(\Subcl(\Sigma))}$ and self-adjoint operators on the relevant Hilbert space. Specifically, it shows that the `regular reals' in $V^{(\Subcl(\Sigma))}$ can always be used to represent the set $SA(\cH)$.

{
Recall that the relation
\[
v(\ck{r})=E^{X}_r
\]
for all $r\in\Q$ sets up a  bijective correspondence between $X\in SH(\cH)$ and $v\in\R^{(\PH)}$
(theorem  \ref{th:q-reals});
in this case, we write $X=\Psi(v)$ and $v=\Phi(X)$.
We then have $\Psi\circ\Phi=\id$ on $SH(\cH)$ and $\Phi\circ\Psi=\id$ on $\RPH$.
Define $\tilde{H}:\RPH\to\RSR$ and $\tilde{G}:\RSR\to\RPH$ 
by $\tilde{H}(v)=H(\Psi(v))$ for all $v\in\R^{(\PH)}$
and  $\tilde{G}(u)=\Phi(G(u))$ for all $u\in \mathbb{R}^{(\Subcl(\Sigma))}_{q}$.
  Then, we obtain
\deqs{
\tilde{G}(u)(\ck{r})&=\Phi(G(u))(\ck{r})=E^{G(u)}_r=\ep(u(\ck{r}))=\alpha(u)(\ck{r}),\\
\tilde{H}(v)(\ck{r})&=H(\Psi(v))(\ck{r})=\de(E^{\Psi(v)}_r)=\de(v(\ck{r}))=\omega(u)(\ck{r})
}
for all $r\in\Q$.  We also have
\deqs{
\alpha\circ\omega(v)&=\tilde{G}\circ\tilde{H}(v)
=\Phi\circ G\circ H\circ \Psi(v)
=v,\\
\omega\circ\alpha(u)
&=\tilde{H}\circ\tilde{G}(u)=H\circ\Psi\circ\Phi\circ G(u)=u
}
for all $u\in \RPH$ and $v\in\RSR$.

Thus, we obtain

\begin{theorem}	
The relations $v=\alpha(u)$ and $u=\omega(v)$ for  $u\in\RSR$ and $v\in\RPH$
set up a  bijective correspondence between $\RSR$ and $\RPH$.
\end{theorem}

}
\section{Conclusion}	\label{Sec_Conclusion}

In conclusion, we take the main results of the paper to be the following.

\begin{itemize}
\item The translation of the orthomodular structure of traditional quantum logic into a new form of distributive and paraconsistent logic that arises naturally in the context of TQT. 
\item The introduction of the paraconsistent negation $*$ into the complete bi-Heyting algebra $\Sub$
to be properly paraconsistent in the sense that $S\And S^{*}=\bot$ only if $S\in \{\bot,\top\}$.
\item The introduction of the commutativity relation into the complete bi-Heyting algebra $\Sub$
expressed by the lattice operation and the paraconsistent negation $*$.
\item The construction of the paraconsistent set theoretic structure $V^{(\Subcl(\Sigma))}$ and the derivation of a $\De_{0}$-theorem transfer scheme that allows us to model major fragments of classical set theory in $V^{(\Subcl(\Sigma))}$. 
\item The translation of the lattice-theoretic operations $\de, \eps$ into higher level set-theoretic maps $\alpha, \omega$, which can subsequently be used to translate ideas and results between TQT and QST.
\item The precise characterisation of the relationship between real numbers in $V^{(\Subcl(\Sigma))}$ and self-adjoint operators on the initial Hilbert space.
\end{itemize}

It goes without saying that this paper only represents a first step in the broader project of unifying TQT and QST into a single overarching formal framework.


\vspace{1cm}

\section*{References}

\vspace{-1cm}

\begin{thebibliography}{9}

\bibitem{Bel11} Bell, J. L., \emph{Set Theory: Boolean-valued Models and Independence
Proofs, Third Edition} (Oxford UP, Oxford, 2005).

\bibitem{Bel08} Bell, J. L., \textit{Toposes and Local Set Theories, An Introduction} 
(Dover, New York, 2008).

\bibitem{BirvNe36} Birkhoff, G. and von Neumann, J., The Logic of Quantum Mechanics, \emph{Annals of Mathematics}, Second Series, 37(4): 823--843 (1936).

\bibitem{Bra89} Brady, R.T., The Non-Triviality of Dialectical Set Theory, in 
Priest, G., Routley, R. and Norman, J. (eds.), \emph{Paraconsistent Logic: Essays on the Inconsistent}: 437--470 (Philosophia, M\"{u}nchen, 1989).

\bibitem{Can13} Cannon, S., \emph{The Spectral Presheaf of an Orthomodular Lattice}, MSc thesis, University of Oxford (2013).

\bibitem{CanDoe16} Cannon, S. and D{\"o}ring, A., A Generalisation of Stone Duality to Orthomodular Lattices, in Ozawa, M. et al. (eds), \textit{Reality and Measurement in Algebraic Quantum Theory}, Proceedings in Mathematics \& Statistics (PROMS) 261: 3--65 
(Springer,  Singapore, 2018).

\bibitem{Doe12} D{\"o}ring, A., Topos-Based Logic for Quantum Systems and Bi-Heyting Algebras, in Chubb, J., Eskandarian, A. and Harizanov, V. (eds.), \emph{Logic and Algebraic Structures in Quantum Computing}, Lecture Notes in Logic 45: 151--173
(Cambridge UP, Cambridge, 2016).

\bibitem{DoeIsh11} D{\"o}ring, A. and Isham, C., What is a Thing?: Topos Theory in the Foundations of Physics, in Coecke, B. (ed.), \emph{New Structures for Physics}, Lecture Notes in Physics 813: 753--940 (Springer, Berlin, 2011).

\bibitem{Dum76} Dummett, M., Is Logic Empirical?, in Lewis, H. D. (ed.), \emph{Contemporary British Philosophy, 4th series}: 45--68 (Allen and Unwin, London 1976).

\bibitem{Eva15} Eva, B., Towards a Paraconsistent Quantum Set Theory, 
in Heunen, C., Selinger, P. and Vicary, J. (eds),
\emph{Proceedings of the 12th International Workshop on Quantum Physics and Logic 
(QPL 2015)}, 
Electronic Proceedings in Theoretical Computer Science (EPTCS) 195: 158--169 (2015).


\bibitem{Eva16} Eva, B., Topos Theoretic Quantum Realism, \emph{British Journal for the Philosophy of Science} 68(4): 1149--1181 (2017). 
DOI: 10.1093/bjps/axv057.

\bibitem{Gib08} Gibbins, P., \emph{Particles and Paradox: The Limits of Quantum Logic} (Cambridge UP, Cambridge, 2008).

\bibitem{Har81} Hardegree, G., Material Implication in Orthomodular (and Boolean) Lattices, \emph{Notre Dame Journal of Formal Logic} 22: 163--182 (1981).

\bibitem{HHLN19} Harding, J., Heunen, C., Lindenhovius, B. and Navara, M.,
Boolean Subalgebras of Orthoalgebras, \emph{Order - A Journal on the Theory of Ordered Sets and its Applications} 36(3): 563--609 (2019).

\bibitem{HarNav11} Harding, J. and Navara, M., Subalgebras of Orthomodular Lattices, \emph{Order - A Journal on the Theory of Ordered Sets and its Applications} 28(3): 549--563 (2011).

\bibitem{Ish97} Isham, C., Topos Theory and Consistent Histories: The Internal Logic of the Set of all Consistent Sets, \emph{International Journal of Theoretical Physics} 36: 785--814 (1997).

\bibitem{IshBut98} Isham, C. and Butterfield, J., A topos perspective on the Kochen-Specker theorem 1: Quantum States as Generalized Valuations, \emph{International Journal of Theoretical Physics} 37: 2669--2733 (1998).

\bibitem{Jec03} Jech, T., \emph{Set Theory: The Third Millennium Edition, revised and expanded}
(Springer, Berlin, 2003).

\bibitem{Joh0203} Johnstone, P., \emph{Sketches of an Elephant, A Topos Theory Compendium, Vols. I, II} (Cambridge UP, Cambridge, 2002/03).

\bibitem{Lib05} Libert, T., Models for a paraconsistent set theory, \emph{Journal of Applied Logic} 3: 15--41 (2005).

\bibitem{LowTar15} L{\"o}we, B. and Tarafder, S., Generalised Algebra-Valued Models of Set Theory, \emph{Review of Symbolic Logic} 8: 192--205 (2015).

\bibitem{McKWeb12} McKubre-Jordens, M. and Weber, Z., Real Analysis in Paraconsistent Logic, \emph{Journal of Philosophical Logic} 41: 901--922 (2012).

\bibitem{McLMoe92} Mac Lane, S. and Moerdijk, I., \textit{Sheaves in Geometry and Logic, A First Introduction to Topos Theory} (Springer, New York, 1994).
	
\bibitem{Put75} Putnam, H., The Logic of Quantum Mechanics, in \emph{Mathematics, Matter and Method}: 174--197 
(Cambridge UP, Cambridge, 1975).

\bibitem{Pri79} Priest, G., Logic of Paradox, \emph{Journal of Philosophical Logic} 8: 219--241 (1979).

\bibitem{Oza04} Ozawa, M.,  Uncertainty Relations for Joint Measurements of Noncommuting Observables, \emph{Physics Letters A} 320: 367--374 (2004).

\bibitem{Oza07} Ozawa, M., Transfer Principle in Quantum Set Theory, \emph{Journal of Symbolic Logic} 72: 625--648 (2007).

\bibitem{Oza16b} Ozawa, M., Quantum Set Theory Extending the Standard Probabilistic Interpretation of Quantum Theory, 
\emph{New Generation Computing} 34: 125--152 (2016).

\bibitem{Oza16} Ozawa, M., Operational Meanings of Orders of Observables Defined Through Quantum Set Theories With Different Conditionals, 
in Duncan, R. and Heunen, C. (eds),
\emph{Proceedings of the 13th International Workshop on Quantum Physics and Logic (QPL2016)},  
Electronic Proceedings in Theoretical Computer Science (EPTCS) 236: 127--144 (2017).

\bibitem{Oza17} Ozawa, M., Orthomodular-Valued Models for Quantum Set Theory, \emph{Review of Symbolic Logic} 10: 782--807 (2017).

\bibitem{RS80} Reed, M. and Simon. B., Methods of Modern Mathematical Physics I: {Functional} Analysis (Revised and Enlarged Edition),
(Academic, New York, 1980).

\bibitem{Sto36} Stone, M.H., The Theory of Representations for Boolean algebras, \emph{Transactions of the 
American Mathematical Society} 40: 37--111 (1936).

\bibitem{Tak74} Takeuti, G., \emph{Two Applications of Logic to Mathematics} (Princeton UP, Princeton, 1974).

\bibitem{Tak81} Takeuti, G., Quantum Set Theory, in Beltrameti, E. G.  and van Fraassen, B.(eds), \emph{Current Issues in Quantum Logic}:  303--322  (Plenum, New York, 1981).

\bibitem{Tit99} Titani, S., Lattice Valued Set Theory, \emph{Archive for Mathematical Logic} 38(6): 395--421 (1999).

\bibitem{Web10} Weber, Z., Transfinite Numbers in Paraconsistent Set Theory, \emph{Review of Symbolic Logic} 3: 71--92 (2010).

\bibitem{Web12} Weber, Z., Transfinite Cardinals in Paraconsistent Set Theory, \emph{Review of Symbolic Logic} 5: 269--293 (2012).

\bibitem{Yin05} Ying, M., A Theory of Computation Based on Quantum Logic (I), 
\emph{Theoretical Computer Science} 344:  134--207 (2005)

\end{thebibliography}
\end{document}